\documentclass[english,12pt]{article}
\usepackage[T1]{fontenc}
\usepackage[utf8]{inputenc}
\usepackage[english]{babel}
\usepackage[pdftex]{graphicx}
\usepackage[width=\textwidth]{caption}
\usepackage[top=24mm, bottom=24mm, left=24mm, right=24mm]{geometry}
\usepackage{array,float, placeins,flafter, longtable,booktabs,pdflscape}
\usepackage{verbatim}
\usepackage{lmodern}
\usepackage{xr}
\usepackage{mathtools}
\usepackage{multirow}
\usepackage{amsfonts,amsmath,amsthm,amssymb,bm,bbm,dsfont}
\allowdisplaybreaks
\usepackage{lmodern}
\usepackage{enumitem}
\usepackage{natbib,color}
\usepackage[pdftex]{graphicx}
\usepackage{subfigure,float,booktabs}
\usepackage{dcolumn}
\newcolumntype{L}{D{.}{.}{2,3}}
\usepackage{pbox}
\usepackage{multirow}
\usepackage{tikz}
\usepackage{caption}
\usepackage{longtable}
\usepackage{booktabs}
\usepackage{tabularx}
\usepackage{dcolumn}
\usepackage{threeparttable}
\usepackage{threeparttablex} 
\usepackage{booktabs}
\usepackage{float} 
\usepackage{setspace}

\usepackage{array}
\usepackage{changepage}
\usepackage{amsmath}
\usepackage{longtable}
\usepackage{mathtools}
\definecolor{winered}{rgb}{0.5,0,0}
\usepackage[pdfstartview=FitH,  bookmarks=true,  bookmarksopen=true,  breaklinks=true,  colorlinks=true,  linkcolor=blue,anchorcolor=my,  citecolor=blue, filecolor=blue,  menucolor=blue,  urlcolor=black]{hyperref}
\numberwithin{equation}{section}
\newtheorem{theorem}{Theorem}[section]
\newtheorem{lemma}[theorem]{Lemma}
\newtheorem{proposition}[theorem]{Proposition}

\theoremstyle{definition}

{\theoremstyle{plain}
	\newtheorem{assumption}{Assumption}}
\definecolor{my}{rgb}{0.05,0.05,0.5}
\definecolor{myBlue}{rgb}{.1,.1,.5}
\definecolor{myGreen}{rgb}{0,.4,0}
\definecolor{myRed}{rgb}{.25,0.15,.5}

\definecolor{my}{rgb}{0.05,0.05,0.5}

\newcommand{\cond}{\displaystyle \stackrel{d}{\longrightarrow}}

\newcommand{\conp}{\stackrel{p}{\longrightarrow}}

\renewcommand{\mathbf}[1]{\textbf{\textit{#1}}}

\newcommand{\Norm}[1]{\mathcal{N}\left(#1\right)}

\newcommand{\E}{\operatorname{E}}
\newcommand{\V}{\operatorname{Var}}

\newcommand{\Rmnum}[1]{\expandafter\@slowromancap\romannumeral #1@}

\usepackage{abstract}
\makeatletter
\g@addto@macro\maketitle{\vspace{-2em}}
\makeatother

\usepackage{setspace}
\newfont{\bbf}{cmbx12 scaled 1435}

\begin{document}
	\title{{Digital Adoption and Cyber Security: An Analysis of Canadian Businesses}\thanks{The authors thank participants at the 58th Annual Conference of the Canadian Economics Association, the 39th Annual Meeting of the Canadian Econometrics Study Group (CESG), and the Bank of Canada Workshop on Understanding and Correcting for Non-Response Bias for their helpful comments. The authors also gratefully acknowledge access to data provided by the RDC and Statistics Canada (Project 21-MAPA YRK-721).}}
	\author{Joann Jasiak\thanks{York University, \texttt{jasiakj@yorku.ca}} \and Peter MacKenzie\thanks{York University, \texttt{petem9@yorku.ca}} \and Purevdorj Tuvaandorj\thanks{York University, \texttt{tpujee@yorku.ca}}}
	\date{\today} 
	
	\maketitle

\begin{abstract}
	This paper examines how Canadian firms balance the benefits of technology adoption against the rising risk of cyber security breaches. We merge data from the 2021 Canadian Survey of Digital Technology and Internet Use and the 2021 Canadian Survey of Cyber Security and Cybercrime to investigate the trade-off firms face when pursuing digitalization to enhance productivity and efficiency, balanced against the potential increase in cyber security risk. The analysis explores the extent of digital technology adoption, differences across industries, the subsequent associations with  efficiency, and associated cyber security vulnerabilities. We build aggregate variables, such as the Business Digital Usage Score and a cyber security incidence variable to quantify each firm's digital engagement and cyber security risk. A survey-weight-adjusted Lasso estimator is employed, and a debiasing method for high-dimensional logit models is introduced to identify the predictors of technological efficiency and cyber risk. The analysis reveals a digital divide linked to firm size, industry, and workforce composition. While rapid expansion of tools such as cloud services or artificial intelligence can raise efficiency, it simultaneously heightens exposure to cyber threats, particularly among larger enterprises. 
\end{abstract}

\newpage

\section{Introduction}

The digital transformation of business operations is reshaping the Canadian economy, offering opportunities for increased productivity and innovation. For businesses, digital adoption is no longer optional; it is essential for maintaining competitiveness in a globalized market. However, with greater reliance on digital tools comes increased exposure to cyber security risks. Understanding this \emph{trade-off}—the efficiency gains from digital adoption versus the vulnerabilities it introduces—is critical for Canadian businesses.

This study examines the extent of digital adoption and cyber security practices among Canadian firms, using a novel dataset that combines information from the 2021 Canadian Survey of Digital Technology and Internet Use (SDTIU) and the 2021 Canadian Survey of Cyber Security and Cybercrime (CSCSC). These surveys provide detailed data on firms’ use of advanced technologies such as cloud computing, artificial intelligence, and enterprise management systems, as well as the frequency, incidence, and severity of cyber security incidents on business operations. To our knowledge, this is the first study to analyze the most recent version of the CSCSC dataset, offering an opportunity to jointly examine digital adoption and cyber security outcomes.

We investigate which types of firms are adopting digital technologies, how adoption varies across industries, and whether these technologies are associated with higher rates of cyber security incidents. We also analyze firm-specific characteristics—such as size, industry, and workforce composition—that shape both technological efficiency and cyber resilience. To quantify digital adoption, we construct the Business Digital Usage Score (BDUS), which enables a comparative analysis of firms’ engagement with digital tools.

To address these questions empirically, we consider a suite of econometric models that link firm characteristics to digital adoption and cyber outcomes. The novelty of our approach lies in introducing high-dimensional logit models estimated via a Lasso-penalized maximum likelihood estimator with survey weights. This allows us to accommodate the large number of firms and the high dimensionality of mostly qualitative explanatory variables in our dataset and their interactions. We use a debiasing method for inference on the selected model's coefficients and establish its asymptotic validity. To assess how closely firms operate relative to their technological usage frontier, we apply stochastic frontier analysis. Additionally, we employ $k$-means clustering to categorize firms by technological efficiency, facilitating the identification of distinct profiles of digital adoption and efficiency.

The literature on digital technology adoption emphasizes its role in improving firm productivity. Larger firms are often better positioned to implement these technologies, benefiting from economies of scale and greater access to resources \citep{Leung2008}. In contrast, smaller firms frequently encounter barriers, including high costs of implementation and limited technical expertise, which can hinder their ability to fully realize the potential benefits of digital adoption. \cite{ferrari2012} demonstrates that industry-specific differences in digital readiness significantly affect adoption rates, while \cite{aghimien2021} highlight the role of regional policies in either enabling or constricting firms’ technological advancement. Together, these studies emphasize the structural factors shaping disparities in digital technology adoption across firms.\par

In Canada, \cite{bilodeau2019}, using data from the 2017 CSCSC, show the widespread reliance of Canadian businesses on digital tools. They report that about 92\% of Canadian businesses used one or more digital technologies or services in 2017, with significant increases in the adoption of websites and social media integration since 2013. In addition, just over one-fifth of Canadian businesses reported being impacted by cyber security incidents that affected their operations, with 54\% noting that these incidents prevented employees from carrying out day-to-day work and about 30\% experiencing additional repair or recovery costs.\par

While digital adoption is associated with higher productivity, it is also associated with increased risks, particularly related to cyber security. The Geneva Association, a leading international think tank of the insurance industry, defines cyber risks as breaches in confidentiality, availability, and data integrity, posing operational threats to firms that increasingly rely on interconnected digital systems \citep{geneva2016}. \cite{cebula2010} expand on this definition, framing cyber risks as disruptions that extend beyond information technology (IT) systems to affect broader business stability. In Canada, the Toronto Public Library system experienced significant disruptions following a cyberattack, while the Nova Scotia Health Department faced operational delays and data breaches during a similar event \citep{toronto_library,NS_health}.\par

The adoption of advanced technologies such as Internet of Things (IoT) and enterprise management systems may be associated with elevated cyber security risks. \cite{blichfeldt2021} highlight that while these systems can improve productivity, their complexity often introduces integration challenges that, if poorly managed, can undermine business operations. Additionally, the widespread use of interconnected devices has expanded the potential attack surface for cybercriminals, requiring firms to invest more heavily in cyber security infrastructure.\par

Remote work adoption during the COVID-19 pandemic accelerated the reliance on digital tools but also introduced new risks and challenges. \cite{hackney2022} find that firms effectively utilizing digital tools during the pandemic demonstrated resilience in maintaining operations under lockdown conditions. However, \cite{aczel2021} and \cite{kitagawa2021} report that the rapid shift to remote work led to heightened employee burnout and increased risks of phishing attacks, underscoring the broader implications of accelerated digital adoption.\par

Cyber insurance is a  tool for mitigating cyber security risks. According to the OECD \citep{oecd2017}, the cyber insurance market doubled in size between 2015 and 2020, fueled by firms’ increasing awareness of cyber threats. However, Fitch Ratings \citep{fitch2021} highlights that high premiums and restrictive coverage policies hinder adoption, especially among smaller firms. Globally, cyber security spending is projected to surpass \$170 billion by 2026 \citep{gartner2021}.

Despite the growing demand for cyber insurance, the unobservable nature 
of cyber security investment creates a moral hazard problem. Specifically, 
once a firm purchases cyber insurance, it may strategically reduce its 
investment in cyber security measures, since the insurer cannot directly 
monitor or verify the level of protection the firm maintains. 
\cite{ahnert2022} argue that firms may therefore prioritize visible 
innovations over less transparent risk mitigation strategies, as clients 
are often unable to directly evaluate cyber security expenditures. This 
dynamic creates a paradox in which firms recognize the increasing risks 
of digital adoption but fail to allocate sufficient resources to mitigate 
them effectively.

While the existing literature has documented the risks, costs, and strategic behavior surrounding cyber security investments, relatively little attention has been paid to the underlying firm-level drivers of adoption and efficiency. To date, the literature has not applied high-dimensional methods of analysis to the study of digital adoption and cyber security, nor has it examined the role of workforce composition or the interaction between technology use and cyber security practices. Furthermore, no prior work has examined the impact of cyber security measures within a stochastic frontier framework.
This paper addresses these gaps by combining detailed firm-level survey data with a novel empirical approach designed to capture the complexity of digital adoption and cyber risk across Canadian businesses.

\bigskip
\noindent The paper is organized as follows. Section \ref{sec: Data} describes the datasets used in the study, the SDTIU and the CSCSC, as well as the construction of the associated scores and variables. Section \ref{sub: DebiasedLasso} outlines the paper's methodological contribution. Section \ref{sec:EmpiricalResults} presents the empirical results, discussing the key predictors of technological efficiency and cyber security vulnerabilities. We conclude in Section \ref{sec:conclusion} with a discussion of the implications of our findings for Canadian businesses and policymakers. The appendix is divided into two parts: Appendix \ref{sec: tech app} provides additional technical details on the survey-weighted debiased logit Lasso model, and Appendix \ref{sec:AppendixQuestions} outlines the questions that comprise the BDUS and the Cyber Security Incidence indicator.

\section{Data Description and Variable Construction}\label{sec: Data}

The empirical analysis is based on a merged dataset from two Statistics Canada surveys: the 2021 SDTIU and the 2021 CSCSC. These surveys were merged using firm size and industry classifications from the North American Industry Classification System (NAICS), enabling an integrated study of digital adoption and cyber security risks among Canadian businesses.\footnote{Methodological differences between the surveys include variations in enterprise size definitions (SDTIU defines small firms as those with 5–49 employees, while CSCSC with 10–49 employees) and potential differences in primary respondents' understanding of their enterprise’s operations. Additionally, the surveys differ in weighted samples of 327,567 and 185,644, respectively, including NAICS and enterprise size requirements.} \par

The SDTIU focuses on digital technology adoption, including metrics such as internet usage, e-commerce participation, and ICT adoption, while the CSCSC examines cyber security practices and the impact of cyber incidents on businesses. The merged dataset includes both qualitative variables, such as the implementation of specific technologies and cyber security measures, and quantitative variables, such as cyber security expenditure and incident-related costs. Together, the surveys cover numerous variables relevant to digital adoption and cyber security. The firms vary in size, ranging from small enterprises with fewer than 10 employees to large corporations with over 500 employees. The industries covered include manufacturing, retail, professional services, and information technology. \par

Both surveys apply stratified sampling by industry and firm size, with survey weights assigned by Statistics Canada to correct for selection probabilities, non-response, and sampling biases relative to each survey's target population of Canadian enterprises. Because the two 
surveys differ in their inclusion criteria and target populations (weighted totals of 327,567 for the SDTIU and 185,644 for the CSCSC), the weights are not common across the two datasets. After merging, we adopt the CSCSC survey weights. Because the CSCSC applies stricter 
inclusion criteria than the SDTIU, the merged sample is effectively restricted to firms within the CSCSC sampling frame. This is confirmed by the post-merge weighted population of 179,657 firms, which closely approximates the CSCSC target population of 185,644 rather than the 
SDTIU target population of 327,567. All descriptive statistics and regression estimates incorporate these weights, which scale the sample to be representative of 179,657 firms meeting the CSCSC survey's industry and firm size inclusion criteria (the weighted population).

The SDTIU achieved a response rate of 73\%, while the CSCSC had a response rate of 65\%. The pre-matched sample sizes were 15,683 enterprises for the SDTIU and 12,216 for the CSCSC. After merging, the final dataset comprises 8,377 enterprises, which constitutes the analysis sample. This sample is not a census of all Canadian firms; rather, through the CSCSC survey weights, it is designed to be representative of 179,657 Canadian firms meeting the CSCSC industry 
and firm size inclusion criteria, which we refer to as the weighted population throughout the paper. All descriptive statistics and regression estimates incorporate these survey weights.

Below, we introduce aggregate measures of digital technology usage: BDUS in Section \ref{subsec: BDUS}, a separate measure of Business Technological Efficiency (based on $k$-means clustering) in Section \ref{subsec: BTE}, which examines whether firms encounter significant challenges in implementing the technologies they adopt, and an indicator for Cyber Security Incidence in Section \ref{subsec: CSI}.

\subsection{Business Digital Usage Score}\label{subsec: BDUS}

The BDUS is a quantitative variable constructed to evaluate the digital engagement of Canadian firms by measuring the technologies they have adopted. It condenses responses from the 2021 SDTIU into a single, interpretable score that reflects the cumulative utilization of digital tools across ten distinct domains, such as cloud computing, digital payment systems, artificial intelligence (AI), smart devices, online sales, government digital connectivity, fiber-optic internet, and company websites (see Appendix \ref{sec:AppendixQuestions} for the exact questions).

Throughout this paper, we use ``digital technologies'' to refer specifically to these ten BDUS domains. In the regression analyses (Section~\ref{sec:EmpiricalResults}), we use ``digital practices'' to describe explanatory variables capturing how firms implement and manage their digital operations (e.g., remote work arrangements, ICT training, advertising methods, open-source technology adoption).

For each of the ten domains, we create a binary variable equal to 1 if the firm reports using that technology. Summing these indicators yields a score between 0 and 10, with higher values indicating more extensive digital engagement. Although it is conceivable that advanced technologies (e.g., AI) have disproportionate associations on firm productivity compared to simpler ones (e.g., a basic website), weighting them by perceived importance would introduce additional assumptions. Firms also have heterogeneous technological needs, some rely heavily on cloud computing for scalable data storage, while others benefit more from online payment systems or data analytics. Therefore uniform summation provides a transparent, easily replicable index of a firm’s overall digital footprint.\par 

A consequence of equal weighting is that the BDUS captures the 
breadth of adoption rather than its sophistication. Two firms may 
receive the same score despite very different adoption profiles: 
one through a collection of basic tools such as a website, social 
media, and fiber optic internet, and another through more advanced 
technologies such as cloud computing, artificial intelligence, and 
big data analytics. To the extent that technologically complex 
tools contribute more to firm productivity, the BDUS does not 
distinguish between these profiles, and alternative weighting 
schemes that account for such differences could affect the 
interpretation of the frontier estimates. That said, any weighting 
scheme would require assumptions about the relative importance of 
each technology, which are likely to vary across industries and 
firm types, and there is no consensus in the literature on how such 
weights should be assigned. \par 

\begin{figure}[h!]
	\centering
	\includegraphics[width=1\linewidth, height=0.3\textheight, keepaspectratio]{"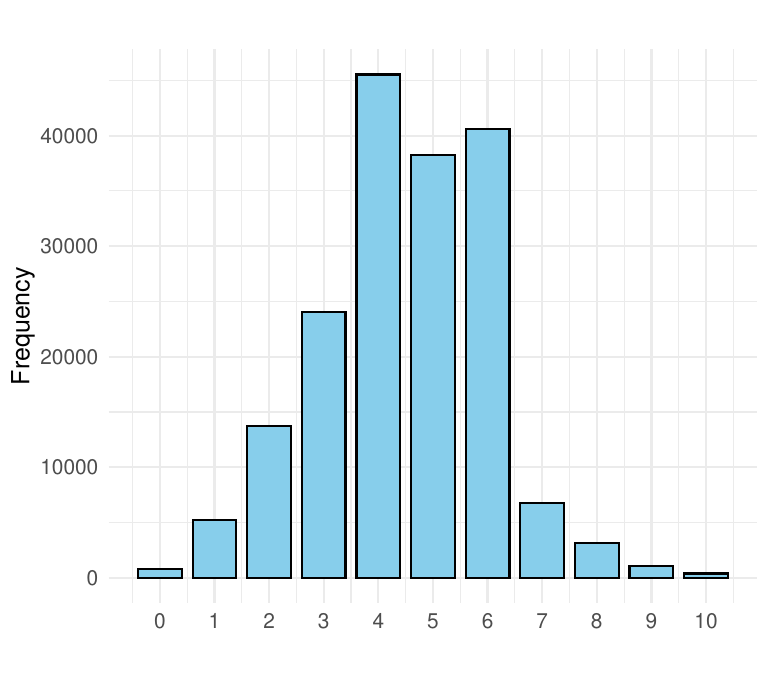"}
	\vspace{-1em}
	\caption{Histogram of Business Digital Usage Scores}
	\label{fig:BDUS_histogram}
\end{figure}

Figure~\ref{fig:BDUS_histogram} shows the distribution of BDUS scores. The majority of firms fall within scores of 3 to 7, suggesting moderate overall adoption levels, with mild peaks at 4 and 6. At the extremes, a small proportion of firms report virtually no digital engagement (score = 0), whereas another small subset achieves near-comprehensive adoption (scores of 9 or 10). This dispersion highlights the heterogeneity of adoption patterns: some businesses implement only a handful of relevant technologies, while others adopt a wide array.\par 

It is important to note that the BDUS does not measure the intensity of usage or the ease of integration. Rather, it provides a concise snapshot of whether certain well-known tools have been adopted in at least a basic form. We use the BDUS to investigate the characteristics of firms that adopt digital technologies through a stochastic frontier model in Section \ref{sec:EmpiricalResults}, measuring how close businesses are to their technological frontier. Additionally, we use the BDUS to assess whether having a more extensive digital profile correlates with cyber risk exposure (Table~\ref{tab:corr_bdus_cyber}).

\subsection{Business Technological Efficiency}\label{subsec: BTE}

While the BDUS captures which digital technologies firms adopt, it does not gauge how well these technologies are integrated. To address this, we construct a measure of Technological Efficiency by applying a $k$-means clustering algorithm to a set of SDTIU survey items about technological implementation challenges. Specifically, these survey items parallel the BDUS domains but ask whether the firm experiences difficulties using each technology, whereas the BDUS questions simply ask if firms use each technology. It is important to note that this measure captures self-reported 
implementation challenges rather than efficiency in a production-theoretic sense. The measure reflects whether firms encounter operational difficulties using adopted technologies, not whether those technologies are deployed in a manner that maximizes output or minimizes cost. The exact questions used in this clustering exercise are listed in Appendix \ref{sec:AppendixQuestions}.

Let $\{z_{i1}, z_{i2}, \dots, z_{i10}\}$ be binary indicators for firm $i$, capturing whether it reports a challenge in each of the ten domains. A response of ``Yes'' indicates an operational problem in using the technology associated with that domain. We apply $k$-means clustering to these ten binary variables and determine that $k=2$ is the optimal cluster count, resulting in two groups: \textit{Digitally Efficient} and \textit{Not Digitally Efficient}. Figure~\ref{fig:tech_problems_clusters} displays the percentage of reported challenges in each domain for these two clusters. Firms in the latter cluster (29\% of the sample) report significantly more issues than those in the former (71\% of the sample). For example, 70.71\% of \textit{Not Digitally Efficient} firms cite difficulties with AI, compared to 23.58\% in the \textit{Digitally Efficient} group. Similarly, 57.55\% of \textit{Not Digitally Efficient} firms encounter challenges with cloud computing, versus only 3.21\% among \textit{Digitally Efficient} firms.

\begin{figure}[h!]
	\centering
	\includegraphics[width=0.9\textwidth]{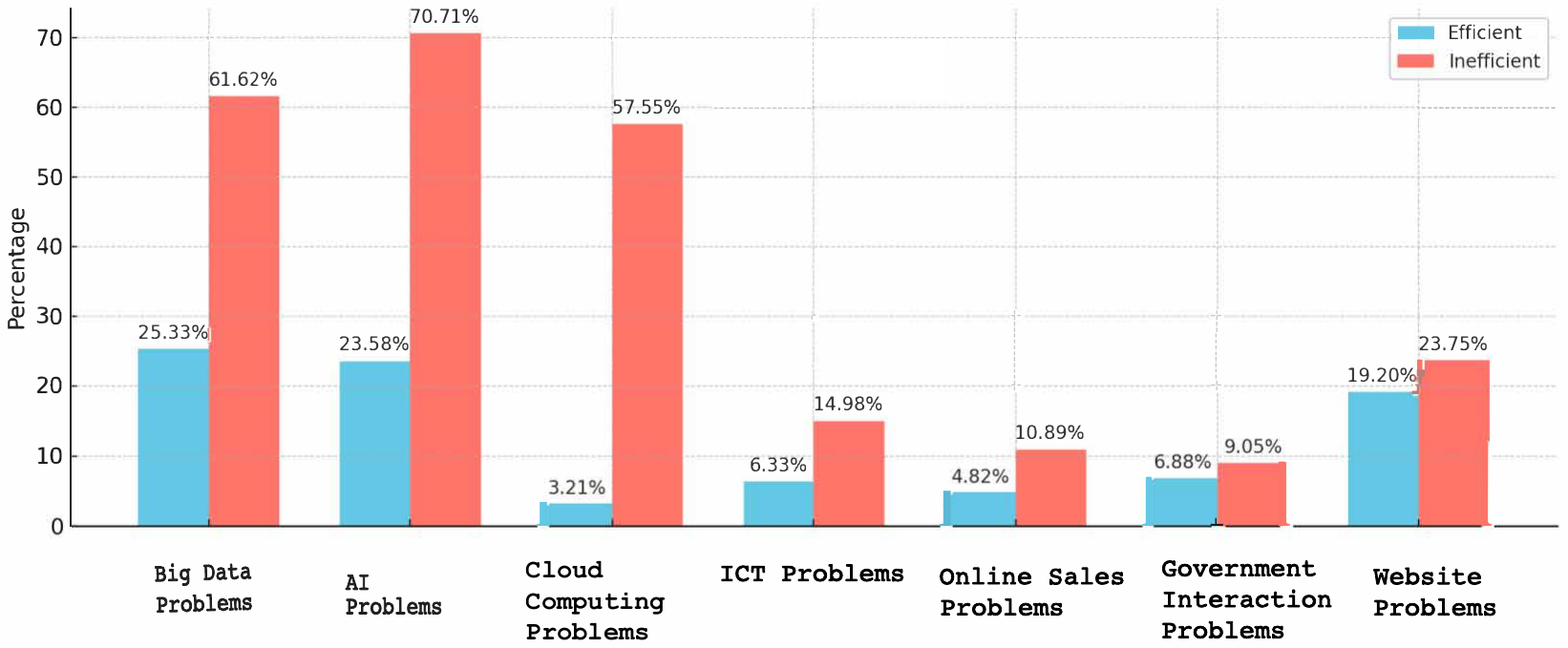}
	\vspace{-5em}
	\caption{Percentage of Technological Problems Between Efficient and Inefficient Clusters}
	\label{fig:tech_problems_clusters}
\end{figure}

The high incidence of challenges among \textit{Not Digitally Efficient} firms does not imply that they fail to adopt these tools. Some businesses report both a high BDUS score (indicating widespread adoption) and frequent operational issues, suggesting partial or suboptimal implementation. Even in the \textit{Digitally Efficient} cluster, a non-trivial share of firms faces difficulties in at least one domain.

This $k$-means classification yields a binary variable, Technological Efficiency, which we use as a dependent variable in one of the logit models in Section~\ref{sec:EmpiricalResults}. Whereas the BDUS measures the extent of adoption, the Technological Efficiency grouping reflects the firm’s ability to use the technology effectively. In this sense, the two measures are complementary: the BDUS indicates how many digital tools a firm adopts, while the logit model using the Technological Efficiency variable reveals whether the adopted technologies are being used efficiently.

We analyze technological efficiency separately from the stochastic frontier model because the challenge-based clustering variables are mechanically related to adoption: firms cannot report implementation challenges for technologies they have not adopted. Including these variables as environmental factors in the SFA would therefore introduce conditioning on a post-outcome variable, potentially biasing efficiency estimates.

\subsection{Cyber Security Incidence}\label{subsec: CSI}

To examine the cyber security challenges faced by Canadian businesses, we construct a set of variables based on survey responses related to the occurrence of cyber security incidents. The primary variable, \textit{Cyber Security Incidence}, indicates whether a business was affected by any cyber security incident in 2021. This binary variable is coded as 1 if the business reported experiencing one or more types of cyber security incidents, and 0 otherwise. These incidents range from theft of assets, business data, or intellectual property to disruptions of business activities. Of the surveyed firms, 22\% reported at least one such incident, while the remaining 78\% reported none.

\begin{figure}[h!]
	\centering
	\includegraphics[width=0.9\textwidth, height=0.4\textheight,]{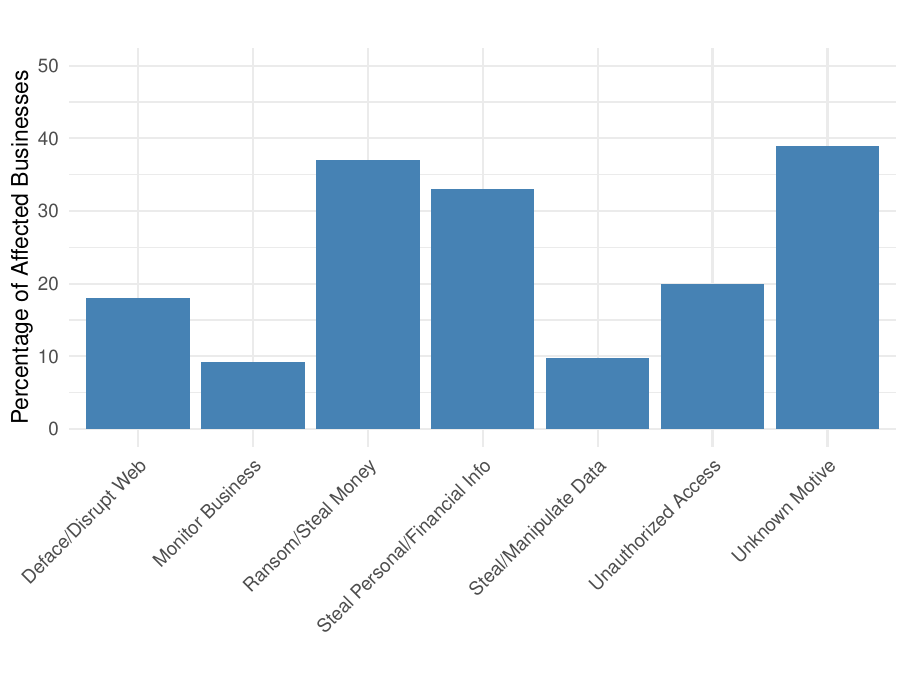}
	\vspace{-0.5em} 
	\caption{Types of Issues Reported by Businesses After Cyber Security Incident (Percentage of Affected Businesses)}
	\label{fig:cyber_impact}
\end{figure}

Figure~\ref{fig:cyber_impact} shows the occurrence of these incident types among the firms that experienced any cyber incident. The most frequently cited issues included incidents to steal money or demand ransom payment (37\%), incidents with an unknown motive (39\%), incidents to steal personal or financial information (33\%), incidents to disrupt or deface the business or web presence (18\%), and incidents to access unauthorised or privileged areas (20\%). Fewer respondents reported incidents to steal or manipulate intellectual property or business data (9.8\%) and incidents to monitor and track business activity (9.2\%). \par

We use this binary \textit{Cyber Security Incidence} variable in subsequent analyses for two primary reasons. It serves as a dependent variable in the cyber security incidence logit model (Section~\ref{sec:EmpiricalResults}), where we identify which firm characteristics and digital adoption practices predict a higher likelihood of experiencing an incident. The \textit{Cyber Security Incidence} variable is also included as a regressor in our stochastic frontier analysis to examine whether having experienced a breach correlates with digital technology usage. The rationale is that a prior cyber incident may be associated with a firms’ subsequent decisions or capabilities regarding digital adoption and security investments.  A full list of cyber incidence types appears in Appendix~\ref{sec:AppendixQuestions}.

\section{Estimation Methodology: Survey Weighted Debiased Lasso}\label{sub: DebiasedLasso}
	
To estimate models predicting technological efficiency and cyber security incidents, we employ a survey-weighted logistic Lasso (hereafter \texttt{svy LLasso}) estimation technique and introduce a debiasing method for inference. This approach is well-suited to our analysis, as it effectively handles high-dimensional data and incorporates a diverse set of variables, including digital practices, cyber security measures, and firm characteristics. Our primary Lasso models include over 50 independent variables, while second-order interaction models expand to more than 200 variables.

To accurately represent the population of Canadian businesses, we adapt the standard logistic Lasso method to incorporate Statistics Canada survey weights. To this end, we define the following general
$\ell_1$-penalized maximum likelihood estimator for a Generalized Linear Model with survey weights, where the \texttt{svy LLasso} estimator is a special case:
	\[
	\hat{\theta} = \operatorname*{argmin}_{\theta=(\alpha, \beta')'\in\mathbb{R}^{p+1}} \left(-L_n(\theta) + \lambda \sum_{j=1}^{p} |\beta_{j}|\right),
	\]
	where:
	\begin{itemize}
	\item \(\theta = (\alpha, \beta')'\) includes the intercept \(\alpha\) and coefficients \(\beta \in \mathbb{R}^p\),
	\item \(L_n(\theta) = -n^{-1}\sum_{i=1}^n w_i g(y_i, x_i'\theta)\) is the survey-weighted log-likelihood, where \(g(y, x'\theta)\) is the negative log-density function (see Appendix \ref{subsec: app2} for a detailed description), \(x_i\) is the regressor vector for firm \(i\), \(y_i\) is the outcome variable, and \(w_i\) is the survey weight,
		\item \( \lambda \sum_{j=1}^{p} |\beta_{j}| \) is the penalty term, with tuning parameter \( \lambda \), enforcing sparsity by shrinking less relevant coefficients to zero.
	\end{itemize}

The methodological contribution of this paper is an extension of 
the existing debiased Lasso literature to the survey-weighted 
setting. The foundational debiased Lasso results of 
\citet{ZhangZhang2014}, \citet{JavanmardMontanari2014}, \citet{Xia2023}, and \citet{ChetverikovSorensen2025}  assume independent 
and identically distributed (i.i.d.) observations with a standard 
unweighted likelihood. Survey weights prevent direct application of 
these results for two reasons. First, the weighted likelihood creates heterogeneous contributions 
across observations, rendering the data independent but not 
identically distributed (i.n.i.d.), which violates the i.i.d.\ 
assumption on which the standard theoretical guarantees for the 
debiased Lasso rest. Second, and more 
critically for inference, the classical information equality between 
the expected Hessian and the variance of the score function breaks 
down under survey weighting: $H(\theta_0) \neq I(\theta_0)$ in 
general. This necessitates a sandwich variance estimator in the 
debiasing step and a corresponding modification of the asymptotic 
normality proof. Proposition~\ref{prop: DB} establishes the asymptotic 
validity of the resulting debiased estimator, characterizing our 
contribution as an extension of \citet{Xia2023} to 
survey-weighted GLMs.

To our knowledge, the only existing work on Lasso estimation with 
survey data is \citet{McConville2017}, who consider model-assisted 
survey regression estimation with the Lasso under finite population 
asymptotics. Our framework differs in two important respects. 
First, we adopt a superpopulation framework common in econometrics, 
as in \citet{Wooldridge2001}, rather than finite population 
asymptotics. Second, we extend the debiasing step to allow 
post-selection inference on parametric functions of the coefficients 
--- including average marginal effects --- in survey-weighted GLMs, 
whereas \citet{McConville2017} focus on linear regression estimation 
without a debiasing step for post-selection inference. Although the 
extension to the survey-weighted setting is relatively 
straightforward given these building blocks, it does not appear to 
have been carried out previously in either the econometrics or the 
survey statistics literature.

For our logit models, the \texttt{svyLLasso} procedure selects the 
variables most predictive of technological efficiency and cyber 
security incidents while explicitly accounting for survey weights. 
This framework has also been applied in \citet{jasiak2023}. In a 
related setting, \citet{JT2023} develop a selective inference 
approach to post-selection inference with survey-weighted penalized 
likelihoods, whereas the present paper adopts a debiasing approach, 
as described below.

\paragraph*{Inference via Debiasing.} For statistical inference on model parameters and average marginal effects (AME) (Appendix \ref{subsec: app1}), we adapt the debiased Lasso method from \cite{ZhangZhang2014} and \cite{JavanmardMontanari2014} to the survey context described above. Given the \(\ell_1\)-penalized maximum likelihood estimator $\hat{\theta}$, the debiasing (DB) method applies a one-step correction:

\begin{equation}
\tilde{\theta} = \hat{\theta} + \hat{H}(\hat{\theta})^{-1} S(\hat{\theta}),
\end{equation}
where $ \hat{H}(\hat{\theta})$ and \( S(\hat{\theta}) \) are the negative Hessian and score function of the weighted log-likelihood function \( L_n(\theta) \).
The adjustment term \( \hat{H}(\hat{\theta})^{-1} S(\hat{\theta}) \) corrects the bias introduced by the \(\ell_1\)-penalized variable selection, enabling reliable inference. This variant, which uses the standard Hessian, follows \cite{Xia2023}, who consider standard GLMs, adapted here for survey weights. We estimate the asymptotic variance of \( n^{1/2} S(\theta_0) \) using a sample information matrix \( \hat{I}(\hat{\theta}) \) (see Appendix \ref{subsec: app2}), where \( \theta_0 \) represents the unknown true values of \( \theta \).

For a nonlinear parameter function \( \rho(\theta) \) (an \( r \times 1 \) vector, possibly \( n \)-dependent), e.g., the AME, we define the debiased estimator:

\begin{equation}\label{eq: rho DB}
	\tilde{\rho} = \rho(\hat{\theta}) + \dot{\rho}(\hat{\theta})' \hat{H}(\hat{\theta})^{-1} S(\hat{\theta}), \quad \dot{\rho}(\theta) = \frac{\partial \rho(\theta)'}{\partial \theta}.
\end{equation}

	\paragraph*{Asymptotic Validity.}
We establish the asymptotic validity of Wald-type inference in the following proposition. To keep the exposition concise, the underlying framework, definitions, and assumptions are provided in Appendix~\ref{sec: tech app}.
	\begin{proposition}[Asymptotic Validity of Survey Debiasing Estimator]\label{prop: DB}
	Let Assumption \ref{A: AsyValid} hold, and assume that 
		\begin{itemize}
			\item \( \lambda = C \sqrt{\frac{\log p}{n}} \) with \( C = O(1) \), \( p \geq 1 \),
			\item \( p^2 / n \to 0 \), and \( m_0 \log p \sqrt{\frac{p}{n}} \to 0 \) as \( n \to \infty \), where \( m_0 \) is the number of non-zero coefficients of $\theta_0$,
		    \item \( \rho(\theta) \) (with fixed \( r < p + 1 \)) is differentiable near \( \theta_0 \) with a locally Lipschitz Jacobian \( \dot{\rho}(\theta) \), and \( \lambda_{\min}(\dot{\rho}(\theta_0)' \dot{\rho}(\theta_0)) > \lambda_l > 0 \), where \( \lambda_{\min} \) denotes the minimum eigenvalue.
		\end{itemize}
		Then:
		\[
		\left( \dot{\rho}(\hat{\theta})' \hat{H}(\hat{\theta})^{-1} \hat{I}(\hat{\theta}) \hat{H}(\hat{\theta})^{-1} \dot{\rho}(\hat{\theta}) \right)^{-1/2} n^{1/2} (\tilde{\rho} - \rho(\theta_0)) \stackrel{d}{\to} \Norm{0, I_r}.
		\]
	\end{proposition}
	
The proof is provided in Appendix \ref{subsec: proof DB}. The order of the tuning parameter $ \lambda$ is standard in the literature \citep{Buhlmann-vandeGeer(2011), Negahban2012, vandeGeer-etal(2014), Hastie2015}. The assumptions on the number of covariates $p$ and model sparsity $m_0$ align with those in \cite{Xia2023}. In particular, the condition $m_0 \log p \sqrt{p/n} \to 0$ is stronger than the condition $m_0 \frac{\log p}{\sqrt{n}} \to 0$ assumed by \cite{vandeGeer-etal(2014)}. However, unlike \cite{vandeGeer-etal(2014)}, no direct assumption is imposed on the sparsity of the inverse Hessian (or information matrix).

The assumption of a locally Lipschitz Jacobian $\dot{\rho}(\theta)$ is slightly stronger than the usual continuous differentiability condition required for testing nonlinear hypotheses (see, e.g., \citet[Section 9]{Newey-McFadden(1994)} and \cite{Hansen(2022b), Hansen(2022)}). Under this assumption, the estimation error arising from the estimation of $\theta_0$ and $\rho(\theta_0)$ becomes negligible.

Using this proposition, we make inference on the parameters $\theta_0$ and $\rho(\theta_0)$. In sum, the debiasing approach above offers a straightforward and robust way to analyze high-dimensional survey data, ensuring both variable selection and valid inference in the presence of survey weights.

  \section{Empirical Results}\label{sec:EmpiricalResults}
  
 The empirical analysis examines the relationship between digital technology adoption and cyber security vulnerabilities among Canadian businesses. We begin by assessing whether broader digital adoption correlates with increased cyber risk (Section~\ref{sub: adopt_risk}) using rank correlations between the BDUS and various security measures. Next, we estimate a Stochastic Frontier Analysis (SFA) model (Section~\ref{sub:SFA}), treating the BDUS as an ``output'' to identify which factors are associated with proximity to a firm's digital usage frontier. We then employ logit models with Lasso selection for two binary outcomes: Technological Efficiency (Section~\ref{sub:efficiency}) and Cyber Security Incidence (Section~\ref{sub: Cyber_incident}), allowing us to determine whether the same firm characteristics or digital practices that are associated with greater adoption also enhance implementation efficiency or increase cyber vulnerability.
 
 The sample sizes for the analyses are as follows: the \texttt{svy LLasso} results for business efficiency are based on a subsample excluding about 2.5\% of firms with a BDUS score less than 2. The cyber security incidence model uses the full merged dataset of 8,377 businesses.
The \texttt{svy LLasso} procedure is implemented using \texttt{R}, where we utilize the default value of the tuning parameter $\lambda$ from the \texttt{R} package \texttt{glmnet} (see Appendix \ref{sec: tech app}).

\subsection{Digital Technology Adoption and Cyber Security} \label{sub: adopt_risk}

We explore whether a  broader digital presence for a firm correlates with heightened cyber risk. Table~\ref{tab:corr_bdus_cyber} reports polychoric and polyserial correlations between the BDUS (an ordinal score from 0 to 10) that measures the amount of technology a firm has adopted and various cyber security measures, including firms cyber security spending (continuous variable), a  binary indicator for whether or not a firm experienced a cyber incident, and whether or not a firm paid ransom as a result of a cyber security breach. We use polychoric correlation for ordinal to binary comparisons and polyserial correlation for the ordinal to continuous comparison.  \par 

The positive and significant correlation between BDUS and both cyber security spending ($\rho=0.156^{***}$) and experiencing an incident ($\rho=0.083^{***}$) suggests that while digitally engaged firms invest more in security, they also face greater exposure to attacks. Firms with a higher BDUS are also more likely to have to pay a ransom as a result of a cyber security breach ($\rho=0.266^{***}$). The absence of any cyber security measures correlates negatively with BDUS ($\rho=-0.068^{***}$), indicating that firms with minimal digital footprints may perceive fewer threats but also forgo basic protective actions.

\begin{table}[h!]
	\small
	\begin{center}
	\caption{Polychoric/Polyserial Correlations Between BDUS and Cyber Security Measures}
	\label{tab:corr_bdus_cyber}
		\begin{tabular}{l r r}
			\toprule
			\textbf{Cyber Security Measure} & \textbf{Correlation} & \textbf{p-value} \\
			\midrule
			Cyber security spending (numeric)   & 0.156 & $<0.001^{***}$ \\
		Experienced a cyber security incident (binary) & 0.083 & $<0.001^{***}$ \\
		Firm implemented cloud storage (binary) & 0.108 & $<0.001^{***}$ \\
		Firm paid ransom (binary)           & 0.266 & $<0.001^{***}$ \\
		Firm implemented no cyber security measures (binary) & $-0.068$ & $<0.001^{***}$\\
			\bottomrule
		\end{tabular}
		\end{center}
\footnotesize{\textit{Notes:} The table reports polychoric correlations for ordinal--binary comparisons and polyserial correlations for the ordinal BDUS and the continuous measure of cyber security spending. All correlations are statistically significant at the 1\% level. Significance levels: *** \(p<0.01\), ** \(p<0.05\), * \(p<0.10\).}
\end{table}

These correlations provide preliminary evidence of a trade-off: as firms adopt more digital tools (higher BDUS), they may both allocate more resources to cyber security and become more frequent targets of attacks. This finding motivates the proceeding empirical analyses.

\subsection{Digital Technology Adoption by Canadian Businesses}\label{sub:SFA}
Section~\ref{subsec: BTE} identified firms that are technologically 
efficient based on their ability to implement adopted digital tools 
without substantial operational frictions. We now examine which firm- 
and industry-level characteristics help businesses achieve such 
efficiency.

To this end, we estimate a stochastic frontier model that 
characterizes each firm's position relative to a ``digital usage 
frontier.'' The use of the BDUS as the dependent variable in this 
framework departs from the traditional production frontier setting, 
where the output is a physical quantity or revenue measure and 
inefficiency reflects shortfalls relative to best-practice 
production. Here, the ``output'' is an index of technology adoption 
rather than a production outcome, and the frontier represents the 
maximal adoption level attainable by a firm given its observable 
characteristics. The inefficiency term therefore captures 
\emph{under-adoption}: the gap between a firm's observed digital 
footprint and what it could feasibly achieve given its size, 
industry, and other characteristics. Firms with a smaller 
inefficiency term lie closer to this latent digital adoption frontier 
and more fully realize their adoption potential given their 
observable attributes.

Although this application is non-standard, it fits within a 
broader literature extending SFA to nontraditional discrete outcomes. 
\citet{Griffiths2014} develop a stochastic frontier model for ordered 
categorical output and explicitly note applications beyond standard 
production settings, including health outcomes, happiness indices, 
university research output, and financial credit ratings --- 
a framing that accommodates index-type dependent variables such 
as the BDUS. \citet{Hwu2021} further extend frontier methods to 
interval-type outcomes, demonstrating the flexibility of the SFA 
framework when the dependent variable is bounded or coarsened rather 
than continuously observed. For count-valued outcomes specifically, 
\citet{FeHofler2013} develop count data stochastic frontier models and 
apply them to the patents--R\&D relationship, providing the most 
direct methodological precedent for our setting. Building on this, 
\citet{FeHofler2020} develop Poisson and negative-binomial frontier 
models in which the departure from the frontier reflects systematic 
underreporting or underutilization, and their framework maps directly 
onto our setting, where firms may adopt fewer technologies than their 
characteristics would predict.

Because the BDUS is a count variable, we follow \citet{FeHofler2020} 
and employ a Poisson--Half-Normal (PHN) stochastic frontier 
specification. Let $y_i \in \{0,1,2,\dots\}$ denote the observed 
BDUS score for firm $i$ with covariates $x_i$. The frontier mean is 
given by
\begin{equation}
	\lambda_i = \exp(x_i'\beta - u_i), 
\end{equation}
where \(u_i \sim \mathcal{N}^{+}(0, \sigma_u^2)\) follows a 
half-normal distribution representing the non-negative inefficiency 
term that captures structural barriers to digital adoption. Firms 
with smaller $u_i$ lie closer to the digital-usage frontier and more 
fully realize their adoption potential.

Conditional on $u_i$, the BDUS score follows a Poisson distribution,
\begin{equation*}
y_i \vert u_i \sim \mathrm{Poisson}(\lambda_i).\footnote{For comprehensive treatments of stochastic frontier 
	models and efficiency measurement, see 
	\citet{Parmeter2014} and \citet{Kumbhakar2015}.}
\end{equation*}
Integrating out $u_i$ yields the mixture likelihood
\begin{equation}
	P(y_i \vert x_i;\theta)
	=
	\int_0^\infty 
	f\left(y_i \vert x_i, u;\theta\right)
	g(u;\theta)du,
	\label{eq:mixlik1}
\end{equation}
where $f(\cdot)$ is the Poisson p.m.f. evaluated at $\lambda_i=\exp(x_i'\beta - u)$ and $g(\cdot)$ is the half-normal density. Since \eqref{eq:mixlik1} has no closed-form solution, we approximate it by Monte Carlo integration:
\begin{equation}
	P(y_i \vert x_i;\theta)
	\approx 
	\frac{1}{H}
	\sum_{h=1}^H 
	f\left(y_i \vert x_i, u_h;\theta\right)\,
	g(u_h;\theta),
	\label{eq:llapprox}
\end{equation}
where $\{u_h\}_{h=1}^H$ are draws from the mixing distribution. Let $w_i>0$ denote the survey sampling weight for observation $i$. The weighted log-likelihood is
\begin{equation}
	L_n(\theta)
	=
	\sum_{i=1}^n
	w_i \log P(y_i \vert x_i;\theta),
	\label{eq:wll}
\end{equation}
and the estimator is defined as
\begin{equation*}
\hat{\theta}
=
\arg\max_{\theta} L_n(\theta).
\end{equation*}

Because survey weights break the (asymptotic) information equality, the Hessian-based standard errors are no longer valid. We therefore report robust (sandwich) standard errors constructed from the outer product of the individual, survey-weighted scores.\footnote{Estimation is carried out in \texttt{R}; the code used for estimation is available upon request.}

\medskip
Before presenting results, we note that the cyber security variables in our model may raise exogeneity concerns if regulatory requirements create mechanical links between security practices and technology adoption. In Canada, there are no federal mandates requiring specific cyber security certifications or training as prerequisites for general technology adoption. The Personal Information Protection and Electronic Documents Act (PIPEDA) and provincial privacy laws require ``appropriate safeguards'' for personal information but do not prescribe specific technologies or certifications. Sector-specific requirements exist in regulated industries: OSFI Guideline B-13 governs technology and cyber risk management for federally regulated financial institutions, and provincial health privacy legislation imposes additional requirements on healthcare providers. We include industry fixed effects to absorb these sector-specific regulatory environments. Given our cross-sectional design, we cannot fully address endogeneity concerns and therefore interpret the cyber security coefficients as associations rather than causal effects.

Table~\ref{tab:PoissonSFA_results} reports the maximum-likelihood estimates from the survey-weighted Poisson frontier model.\footnote{Although not reported here, we also estimated a linear stochastic frontier model with half-normal inefficiency using the transformed dependent variable $\log(\mathrm{BDUS}+1)$ to accommodate zero values. The results are qualitatively similar to those reported in Table~\ref{tab:PoissonSFA_results}.}
 Explanatory variables are organized into three categories: \emph{Firm characteristics}, \emph{Digital technologies}, and \emph{Cyber security measures}. 

\begin{table}[h!]
	\small
	\begin{center}
		\caption{Poisson Stochastic Frontier Model for BDUS}
		\label{tab:PoissonSFA_results}
		\begin{tabular}{l D{.}{.}{3} D{.}{.}{3} D{.}{.}{3} D{.}{.}{3}}
			\toprule
			\multicolumn{1}{l}{\textbf{Variable}} &
			\multicolumn{1}{c}{\textbf{Coef.}} &
			\multicolumn{1}{c}{\textbf{Std.\ Error}} &
			\multicolumn{1}{c}{\textbf{t-value}} &
			\multicolumn{1}{c}{\textbf{p-value}}\\
			\midrule
			\quad Intercept                           & 1.211 & 0.018 & 65.498 & <0.001^{***} \\
			\multicolumn{5}{l}{\textit{Firm characteristics}}\\
			\quad Medium-sized firm                   & 0.072 & 0.014 & 5.298  & <0.001^{***} \\
			\quad Large-sized firm                    & 0.170 & 0.016 & 10.454 & <0.001^{***} \\
			\quad Working from home                   & 0.183 & 0.016 & 11.491  & <0.001^{***} \\
			\quad Mining/Utilities/Construction       & -0.178 & 0.032 & -5.627 & <0.001^{***} \\
			\quad Manufacturing                       & 0.040 & 0.014 & 2.882  & 0.004^{***}   \\
			\quad Wholesale/Retail/Transport          & -0.073 & 0.024 & -3.057 & 0.002^{***}   \\
			\quad Education/Health                    & -0.060 & 0.018 & -3.331 & 0.001^{***} \\
			\quad Arts/Accommodation/Food             & 0.155 & 0.020 & 7.719  & <0.001^{***} \\
			\quad Other services                      & -0.032 & 0.025 & -1.258  & 0.209     \\
			\multicolumn{5}{l}{\textit{Digital practices}}\\
			\quad Blockchain                          & 0.075 & 0.036 & 2.058  & 0.040^{**}    \\
			\quad Open source technologies            & 0.111 & 0.015 & 7.656  & <0.001^{***} \\
			\quad Client information management       & 0.229 & 0.015 & 15.189 & <0.001^{***} \\
			\quad Paid advertising                    & 0.055 & 0.016 & 3.487  & <0.001^{***}   \\
			\quad Free advertising                    & 0.113 & 0.016 & 6.882  & <0.001^{***} \\
			\quad Firm provides ICT training          & 0.039 & 0.026 & 1.511  & 0.131     \\
			\quad Sales-related problems              & 0.089 & 0.032 & 2.791  & 0.005^{***}   \\
			\multicolumn{5}{l}{\textit{Cyber security measures}}\\
			\quad Gender in ICT roles (\% female)     & 0.002 & 0.000 & 7.820  & <0.001^{***} \\
			\quad Gender in cyber security (\% female)& -0.000 & 0.000 & -0.835 & 0.404     \\
			\quad Cyber security certification        & 0.032 & 0.016 & 1.971  & 0.049^{**}     \\
			\quad Cyber security practices            & -0.023 & 0.017 & -1.326 & 0.185     \\
			\quad Cyber security training             & 0.037 & 0.018 & 2.005  & 0.045^{**}    \\
			\quad Cyber security insurance            & -0.026 & 0.016 & -1.662 & 0.096^{*}     \\
			\quad Cyber security incidents            & 0.012 & 0.018 & 0.654  & 0.513     \\
			\multicolumn{5}{l}{\textit{Variance parameter:}}\\
			\quad $\log(\sigma)$         & -8.479 & 0.019 & -444.730 & <0.001^{***} \\[4pt]
			\multicolumn{5}{l}{\quad Test for $H_{0}:\sigma = 0$:\quad Wald statistic $=52.449$,\; $p<0.001^{***}$} \\
			\bottomrule
		\end{tabular}
	\end{center}
\footnotesize{\textit{Notes:} All numeric values are rounded to three decimals. Reference categories are \emph{Small firms} (firm size) and \emph{Professional services} (industry). All explanatory variables are indicator variables unless otherwise specified. A positive coefficient indicates that the variable moves the firm closer to its digital usage frontier. Reported standard errors are robust (sandwich) estimates based on the outer product of the individual, survey-weighted scores, and the $p$-values for the coefficients are computed using the standard normal approximation. The $p$-value for the Wald test of $H_{0}:\sigma=0$ is calculated using the $0.5\chi^2_{0}+0.5\chi^2_{1}$ mixture distribution. Significance levels: *** $p<0.01$, ** $p<0.05$, * $p<0.10$. The model is estimated using $H=10{,}000$ simulation draws. 
}
\end{table}

The result shows that several firm characteristics and digital practices are strongly associated with proximity to the digital-usage frontier. \emph{Medium} ($0.072^{***}$) and \emph{Large} ($0.170^{***}$) firms both exhibit higher BDUS levels and perform significantly better than small businesses, while firms adopting remote-work arrangements (\emph{Working from home}, $0.183^{***}$) also align positively with digital adoption. Among digital practices, \emph{Open source solutions} ($0.111^{***}$), \emph{Client information management} ($0.229^{***}$), and \emph{Online advertising} (\emph{Paid}: $0.055^{**}$; \emph{Free}: $0.113^{***}$) are consistently associated with higher BDUS levels.\par 

Certain industry sectors register negative coefficients: \emph{Mining/Utilities/Construction} 

\noindent($-0.178^{***}$), \emph{Wholesale/Retail/Transport} ($-0.073^{**}$), and \emph{Education/Health} ($-0.060^{***}$). In contrast, \emph{Manufacturing} ($0.040^{***}$) and \emph{Arts/Accommodation/Food} ($0.155^{***}$) both exceed the \emph{Professional services} benchmark.\par 

For cyber security measures, the share of \emph{female employees in ICT roles} ($0.002^{***}$) is the only strongly significant predictor, indicating that a higher percentage of women in ICT is associated with greater digital adoption. The test of the variance-parameter nullity, $\sigma=0$, is strongly rejected, indicating that systematic inefficiency dominates random noise. This implies that much of the under-adoption of digital tools is associated with structural constraints rather than chance.
\par

These results imply that resource capacity (firm size), remote-work arrangements, and the availability of specific digital tools or managerial practices (e.g., open source solutions, client information management) are associated with greater proximity to the feasible digital-usage frontier for Canadian firms. Industries such as Construction and Education/Health appear to face structural adoption barriers, while Manufacturing and consumer-facing sectors integrate digital tools more readily. Although extensive digital adoption may heighten cyber risks (Section~\ref{sub: adopt_risk}), most of the cyber security variables analyzed here do not notably influence how far along a firm is on the adoption curve.

\subsection{Technological Efficiency of Canadian Businesses}\label{sub:efficiency}

The SFA analysis in Section~\ref{sub:SFA} measures the extent of digital adoption but does not capture how effectively these tools are utilized. A firm may adopt many technologies yet struggle to integrate them, or vice versa. To address this distinction, we introduce a binary variable for \textit{Technological Efficiency}, derived through $k$-means clustering. This clustering algorithm categorizes firms based on whether they report few or many operational difficulties in using adopted technologies. A firm is classified as ``Technologically Efficient'' if it experiences relatively few challenges across multiple domains (see Section~\ref{sec: Data} for details). We restrict the sample to businesses with BDUS $\geq 2$, ensuring a baseline level of digital engagement prior to evaluating efficiency.

Table~\ref{tab:debiased_lasso_results} presents the \texttt{svy LLasso} results for the probability of a firm being Technologically Efficient. Again, the independent variables are grouped into three categories: \textit{Firm characteristics}, \textit{Digital technologies}, and \textit{Cyber security measures}. A positive coefficient indicates that the variable increases the likelihood of efficient technology use, while a negative coefficient suggests a reduced likelihood. The final two columns report the AMEs and their corresponding $p$-values.

{\small
	\begingroup
	\begin{singlespace}
		\begin{longtable}{l D{.}{.}{3} D{.}{.}{3} D{.}{.}{3} D{.}{.}{3} D{.}{.}{3}}
			\caption{Debiased Logit Lasso Estimation Results for Technological Efficiency\label{tab:debiased_lasso_results}}\\
			\toprule
			\multicolumn{1}{l}{\textbf{Variables}}  
			& \multicolumn{1}{c}{\texttt{svy LLasso}}   
			& \multicolumn{1}{c}{$\tilde{\theta}^{DB}$} 
			& \multicolumn{1}{c}{$p$-value} 
			& \multicolumn{1}{c}{$\widetilde{\mathrm{AME}}^{DB}$} 
			& \multicolumn{1}{c}{$p$-value}\\
			\midrule
			\endfirsthead

			\toprule
			\multicolumn{1}{l}{\textbf{Variables}}  
			& \multicolumn{1}{c}{\texttt{svy LLasso}}   
			& \multicolumn{1}{c}{$\tilde{\theta}^{DB}$} 
			& \multicolumn{1}{c}{$p$-value} 
			& \multicolumn{1}{c}{$\widetilde{\mathrm{AME}}^{DB}$} 
			& \multicolumn{1}{c}{$p$-value}\\
			\midrule
			\endhead
			
			\midrule
			
			\endfoot
			
			\bottomrule
			\endlastfoot

			Intercept
			& -0.060 &  0.105 & 0.606       
			&  0.122 & 0.606       \\
						[3pt]
			\multicolumn{6}{l}{\textit{Firm characteristics}}\\[3pt]
			
			Medium firm
			&  0.551 &  0.856 & <0.001^{***}   
			&  0.158 & <0.001^{***}   \\
			
			Large firm
			&  0.366 &  1.237 & <0.001^{***}   
			&  0.102 & <0.001^{***}   \\
			
			Remote work
			&  0.509 &  0.664 & <0.001^{***}   
			&  0.139 & <0.001^{***}   \\
			
			Female in ICT roles (1--20\%)
			&   .    &  0.486 & 0.254       
			&  0.023 & 0.720       \\
			
			Female in ICT roles (21--40\%)
			&   .    &  0.150 & 0.688       
			&  0.139 & 0.153       \\
			
			Female in ICT roles (41--60\%)
			&  0.211 &  1.066 & 0.087^{*}      
			&  0.123 & 0.379       \\
			
			Female in ICT roles ($>\,$60\%)
			&   .    &  0.915 & 0.273       
			& -0.032 & 0.214       \\
			
			Foreign market
			&   .    &  0.041 & 0.873       
			&  0.031 & 0.566       \\
			
			Mining/Utilities/Construction
			& -0.772 & -1.108 & <0.001^{***}   
			&  0.020 & 0.458       \\
			
			Manufacturing
			&   .    &  0.130 & 0.411       
			& -0.065 & 0.044^{**}     \\
			
			Wholesale/Retail/Transport
			& -0.300 & -0.416 & 0.024^{**}     
			& -0.030 & 0.392       \\
			
			Education/Health
			&   .    & -0.193 & 0.354       
			&  0.113 & 0.006^{***}     \\
			
			Arts/Accommodation/Food
			&  0.256 &  0.808 & 0.002^{***}    
			&  0.146 & <0.001^{***}    \\
			
			Other services
			&  0.524 &  1.086 & <0.001^{***}   
			&  0.146 & <0.001^{***}    \\
			[6pt]
			
			\multicolumn{6}{l}{\textit{Digital practices}}\\[3pt]
			
			Blockchain usage
			&  0.086 &  1.066 & 0.153       
			&  0.046 & 0.428       \\
			
			ICT training
			&  0.045 &  0.314 & 0.374       
			&  0.070 & 0.327       \\
			
			Online orders
			&   .    & -0.211 & 0.166       
			&  0.006 & 0.884       \\
			
			AI
			&   .    &  0.205 & 0.520       
			&  0.264 & <0.001^{***}   \\
			
			IoT
			&  1.765 &  1.971 & <0.001^{***}   
			&  0.156 & <0.001^{***}   \\
			
			Computer network
			&  0.661 &  0.989 & <0.001^{***}   
			& -0.001 & 0.968       \\
			
			Customer relationship management
			&   .    & -0.008 & 0.965       
			& -0.141 & <0.001^{***}   \\
			
			Electronic data interchange
			& -0.680 & -0.883 & <0.001^{***}   
			& -0.101 & 0.007^{**}     \\
			
			Enterprise resource planning
			&   .    & -0.636 & 0.005^{***}     
			&  -0.168 & 0.050^{**}      \\
			
			Big data usage
			&  0.365 &  1.350 & 0.019^{**}     
			&  0.076 & 0.012^{**}     \\
			
			Open source technologies
			&  0.288 &  0.519 & 0.007^{***}     
			& -0.022 & 0.730       \\
			
			Advertising
			& -0.327 & -0.440 & 0.010^{**}     
			&  0.110 & <0.001^{***}   \\
			
			Free advertising
			&  0.431 &  0.740 & <0.001^{***}   
			&  0.073 & 0.025^{**}     \\
			
			Website
			&  0.399 &  0.454 & 0.010^{**}     
			& -0.006 & 0.852       \\
			
			Company apps
			&   .    & -0.041 & 0.839       
			& -0.087 & <0.001^{***}   \\
			
			Social media
			& -0.287 & -0.587 & <0.001^{***}   
			& -0.027 & 0.240       \\
			
			Fiber optic
			&   .    & -0.179 & 0.198       
			& -0.046 & 0.113       \\
			
			Online sales
			& -0.064 & -0.300 & 0.083^{*}      
			& -0.013 & 0.607       \\
			
			Client information management
			&  0.064 &  0.094 & 0.498       
			& -0.068 & 0.019^{**}     \\
			[6pt]
			
			\multicolumn{6}{l}{\textit{Cyber security measures}}\\[3pt]
			
			Female in cyber security roles (1--20\%)
			&   .   & -0.141 & 0.708       
			& -0.012 & 0.815       \\
			
			Female in cyber security roles (21--40\%)
			&   .   & -0.077 & 0.799       
			&  0.030 & 0.443       \\
			
			Female in cyber security roles (41--60\%)
			&   .   &  0.202 & 0.391       
			&  0.049 & 0.166       \\
			
			Female in cyber security roles ($>\,$60\%)
			&  0.062 &  0.329 & 0.123       
			& -0.019 & 0.487       \\
			
			Cyber security employees (1--2)
			&   .    & -0.127 & 0.445       
			& -0.003 & 0.946       \\
			
			Cyber security employees (3+)
			&   .    & -0.020 & 0.941       
			& -0.014 & 0.622       \\
			
			Cyber security insurance
			&   .    & -0.090 & 0.590       
			&  0.014 & 0.537       \\
			
			Employee monitoring
			&   .    & -0.087 & 0.574       
			& -0.184 & <0.001^{***}   \\
			
		\end{longtable}
\noindent\footnotesize{\textit{Notes:} All numeric values are rounded to three decimals. $\tilde{\theta}^{DB}$ and $\widetilde{\mathrm{AME}}^{DB}$ denote the debiased logit Lasso coefficient and AME estimates, respectively. Significance levels: *** $p<0.01$, ** $p<0.05$, * $p<0.10$. Reference categories: Small firm, 0\% female in ICT roles, 0\% female in cyber security roles, 0 cyber security employees, and Industry: Professional services.}
\end{singlespace}
\endgroup
}
\bigskip
Table~\ref{tab:debiased_lasso_results} shows that several firm characteristics are significantly associated with higher probability of using digital tools efficiently.  \emph{Medium} ($\widetilde{\mathrm{AME}}^{DB} =0.158^{***}$) and \emph{Large} ($\widetilde{\mathrm{AME}}^{DB} =0.102^{***}$) firms exhibit higher AMEs relative to \emph{Small} firms. \emph{Remote work} ($\widetilde{\mathrm{AME}}^{DB} =0.139^{***}$) is also associated with increased efficiency. Among digital practices, the \emph{IoT} ($\widetilde{\mathrm{AME}}^{DB} =0.156^{***}$) and \emph{Big data usage} ($\widetilde{\mathrm{AME}}^{DB} =0.076^{**}$) both show strong positive associations, while advertising efforts, such as \emph{Free advertising} ($\widetilde{\mathrm{AME}}^{DB} =0.073^{**}$), also correlate with efficient usage. For industries, \emph{Arts/Accommodation/Food} ($\widetilde{\mathrm{AME}}^{DB} =0.146^{***}$) and \emph{Education/Health} show a positive marginal effect ($\widetilde{\mathrm{AME}}^{DB} =0.113^{***}$) compared to \emph{Professional services}. \par 

Firms endowed with more resources, due to scale (medium or large size) or operational flexibility (remote work) can allocate staff and capital to integrate digital systems more effectively. Real-time connectivity from IoT and big data appears to reinforce structured work flows, while advertising activities may align with better-organized digital platforms. Consumer-facing industries, such as Arts/Accommodation/Food may capitalize on digital marketing tools or consumer facing technologies such as reservation systems more readily than sectors facing heavier regulatory or operational barriers. \par 

The variables associated with reduced likelihood of efficient digital implementation include
\emph{Electronic data interchange} ($\widetilde{\mathrm{AME}}^{DB} =-0.101^{***}$) and \emph{Enterprise resource planning} ($\widetilde{\mathrm{AME}}^{DB} =-0.168^{**}$). The AME values suggest that more complex systems can pose challenges with technological adoption.  The \emph{Manufacturing} industry ($\widetilde{\mathrm{AME}}^{DB} =-0.065^{**}$) has a negative association based on AMEs, while certain cyber security variables, such as \emph{Employee monitoring} ($\widetilde{\mathrm{AME}}^{DB} =-0.184^{***}$), also correlate negatively with Technological Efficiency. \par

These negative or insignificant AMEs point to organizational or regulatory factors that can undermine digital adoption benefits. Complex software solutions (EDI, ERP) often require robust training and IT resources; without sufficient support, firms may experience integration hurdles. Industries like Mining/Utilities/Construction may face specialized work flows or regulatory strictures that may obstruct rapid digital platform adoption. Certain cyber security practices (employee monitoring) can introduce procedural friction or negative employment sentiment that overshadows efficiency gains if not carefully managed.

These findings carry several policy implications. First, programs supporting digital adoption should account for firm size: while small firms may require assistance with initial adoption, medium and large firms benefit from resources targeting effective implementation. Second, the positive association between remote work and technological efficiency suggests that policies facilitating flexible work arrangements may yield collateral benefits for digital integration. Third, the negative associations with complex enterprise systems (EDI, ERP) indicate that adoption support programs should include implementation assistance and training, not merely subsidies for technology acquisition. Finally, the negative association between employee monitoring and efficiency suggests that heavy-handed digital surveillance practices may undermine the productivity gains that digital adoption is intended to deliver.

\subsection{Cyber Security Incidence} \label{sub: Cyber_incident}

The determinants of \textit{Cyber Security Incidence} variable are analyzed using a binary dependent variable introduced in Section~\ref{sec: Data}. The \textit{Cyber Security Incidence} variable equals 1 if a firm reported experiencing at least one cyber security incident during 2021, and 0 otherwise; among the surveyed firms, 18.0\% reported an incident. Cyber incidents encompass a range of adverse events, including theft of business assets, data breaches, disruptions to business activities, intellectual property losses, and other cyber-related issues. Although a large number of independent variables are included in the analysis, relatively few emerge as statistically significant predictors, indicating that cyber risk is shaped by a limited subset of factors. \par 

{\small
\begin{singlespace}
	\begin{longtable}{l D{.}{.}{3} D{.}{.}{3} D{.}{.}{3} D{.}{.}{3} D{.}{.}{3}}
		\caption{\texttt{svy LLasso} Results for Cyber Security Incidence \label{tab:debiased_lasso_results_additional}}\\
		\toprule
		\multicolumn{1}{l}{\textbf{Variables}}  
		& \multicolumn{1}{c}{\texttt{svy LLasso}} 
		& \multicolumn{1}{c}{$\tilde{\theta}^{DB}$} 
		& \multicolumn{1}{c}{$p$-value} 
		& \multicolumn{1}{c}{$\widetilde{\mathrm{AME}}^{DB}$} 
		& \multicolumn{1}{c}{$p$-value} \\
		\midrule
		\endfirsthead

		\toprule
		\multicolumn{1}{l}{\textbf{Variables}}  
		& \multicolumn{1}{c}{\texttt{svy LLasso}} 
		& \multicolumn{1}{c}{$\tilde{\theta}^{DB}$} 
		& \multicolumn{1}{c}{$p$-value} 
		& \multicolumn{1}{c}{$\widetilde{\mathrm{AME}}^{DB}$} 
		& \multicolumn{1}{c}{$p$-value} \\
		\midrule
		\endhead
		
		\bottomrule
		\endfoot
		
		\bottomrule
		\endlastfoot
		
		Intercept                    & -1.904 & -2.989 & <0.001^{***} &   .   &   .        \\[4pt]
				\multicolumn{6}{l}{\textit{Firm characteristics}}\\[4pt]
		Medium firm                    &.        &  0.047 &  0.716    &  0.007 &  0.723    \\
		Large firm                     &.       &  0.370 &  0.038^{**}  &  0.057 &  0.030^{**}  \\
		Remote work                    &.        &  0.104 &  0.476    &  0.015 &  0.488    \\
		Female in ICT roles (1--20\%) &.        & -0.325 &  0.236    & -0.042 &  0.291    \\
		Female in ICT roles (21--40\%)&.        &  0.389 &  0.178    &  0.060 &  0.155    \\
		Female in ICT roles (41--60\%)&.        &  0.123 &  0.781    &  0.018 &  0.783    \\
		Female in ICT roles ($>\,$60\%) &.        & -0.305 &  0.660    & -0.040 &  0.694    \\
		Mining/Utilities/Construction &.        &  0.728 &  0.002^{***} &  0.119 &  0.001^{***} \\
		Manufacturing                 &.        &  0.464 &  0.002^{***} &  0.071 &  0.001^{***} \\
		Wholesale/Retail/Transport    &.        &  0.398 &  0.034^{**}  &  0.059 &  0.032^{**}  \\
		Education/Health              &.        &  0.058 &  0.771    &  0.008 &  0.777    \\
		Arts/Accommodation/Food       &.        &  0.235 &  0.394    &  0.034 &  0.394    \\
		Other services                &.        &  0.484 &  0.083^{*}   &  0.075 &  0.066^{*}   \\[6pt]
		
		\multicolumn{6}{l}{\textit{Digital practices}}\\[4pt]
		Blockchain usage              &.        &  0.202 &  0.667    &  0.030 &  0.663    \\
		ICT training                  &.        &  0.361 &  0.168    &  0.055 &  0.151    \\
		Online orders                 &.        &  0.032 &  0.835    &  0.005 &  0.841    \\
		AI       &.        & -0.289 &  0.207    & -0.038 &  0.255    \\
		IoT            &.        &  0.148 &  0.336    &  0.021 &  0.345    \\
		Computer network              &.        &  0.016 &  0.909    &  0.002 &  0.912    \\
		Customer relationship management
		&        &  0.129 &  0.419    &  0.019 &  0.426    \\
		Electronic data interchange   &.        & -0.176 &  0.293    & -0.024 &  0.323    \\
		Enterprise resource planning  &.        &  0.137 &  0.477    &  0.020 &  0.480    \\
		Big data usage                &.        & -0.414 &  0.307    & -0.053 &  0.374    \\
		Open source technologies      &.        & -0.102 &  0.528    & -0.014 &  0.550    \\
		Confidential cloud            &  0.161 &  0.466 & <0.001^{***} &  0.067 & 0.002^{***}  \\
		Personal device               &.        &  0.222 &  0.120    &  0.031 &  0.141    \\
		VPN                           &.        &  0.003 &  0.986    &  0.000 &  0.986    \\
		Payment services              &.        & -0.123 &  0.663    & -0.017 &  0.682    \\
		Client information management &.        &  0.082 &  0.569    &  0.012 &  0.582    \\
		Website                       &.        & -0.098 &  0.616    & -0.014 &  0.624    \\
		Company apps                  &.        &  0.112 &  0.556    &  0.016 &  0.563    \\
		Social media                  &.        & -0.173 &  0.255    & -0.025 &  0.265    \\
		Online sales                  &.        &  0.072 &  0.663    &  0.010 &  0.673    \\[6pt]
		
		\multicolumn{6}{l}{\textit{Cyber security measures}}\\[4pt]
		Anti malware                  &.        & -0.236 &  0.323    & -0.034 &  0.327    \\
		Web security                  &.        & -0.115 &  0.502    & -0.016 &  0.517    \\
		Email security                &.        &  0.272 &  0.271    &  0.037 &  0.304    \\
		Network security              &.        &  0.087 &  0.691    &  0.012 &  0.704    \\
		Data security                 &.        &  0.164 &  0.362    &  0.023 &  0.374    \\
		POS security                  &.        & -0.115 &  0.502    & -0.016 &  0.522    \\
		Software security             &.        & -0.251 &  0.171    & -0.034 &  0.198    \\
		Hardware security             &.        &  0.203 &  0.259    &  0.029 &  0.266    \\
		Password security             &  0.159 &  0.267 &  0.143    &  0.038 &  0.157    \\
		Access security               &.        &  0.008 &  0.963    &  0.001 &  0.964    \\
		Female in cyber security roles (1--20\%)
		&.        & -0.042 &  0.904    & -0.006 &  0.908    \\
		Female in cyber security roles (21--40\%)
		&.        &  0.062 &  0.819    &  0.009 &  0.823    \\
		Female in cyber security roles (41--60\%)
		&.        &  0.188 &  0.434    &  0.028 &  0.434    \\
		Female in cyber security roles ($>\,$60\%)
		&.        & -0.424 &  0.054^{*}   & -0.056 &  0.082^{*}    \\
		
		Cyber security employees (1--2)
		&   .    & 0.527 & 0.061^{*}     
		& 0.074 & 0.070^{*}     \\
		
		Cyber security employees (3+)
		&   .    & 0.059 & 0.737       
		& 0.008 & 0.743       \\
		
		Cyber security insurance       &.        & -0.321 &  0.050^{**}  & -0.043 &  0.073^{*}   \\
		Cyber consultant               &.        & -0.018 &  0.941    & -0.003 &  0.943    \\
		Cyber information              &  0.295 &  0.214 &  0.450    &  0.030 &  0.459    \\
		Cyber training                 &.        &  0.195 &  0.254    &  0.028 &  0.259    \\
		Cyber policy                   &.        & -0.108 &  0.507    & -0.015 &  0.527    \\
		Cyber practice                 &  0.051 &  0.254 &  0.331    &  0.036 &  0.346    \\
		Employee monitoring            &.        &  0.242 &  0.122    &  0.035 &  0.123    \\
		Risk assessment                &.        &  0.199 &  0.219    &  0.029 &  0.226    \\
		Cyber team: white employees only          &.        & -0.049 &  0.785    & -0.007 &  0.794    \\
		Cyber team: minority employees only           &.        &  0.424 &  0.152    &  0.065 &  0.131    \\
		Cyber certification            &.        &  0.096 &  0.570    &  0.014 &  0.580    \\
		No cyber security measures     &.        & -0.179 &  0.634    & -0.024 &  0.657    \\
\end{longtable}
\noindent\footnotesize{\textit{Notes:} All numeric values are rounded to three decimals. $\tilde{\theta}^{DB}$ and $\widetilde{\mathrm{AME}}^{DB}$ denote the debiased logit Lasso coefficient and AME estimates, respectively. Significance levels: *** $p<0.01$, ** $p<0.05$, * $p<0.10$. Reference categories: Small firm, 0\% female in ICT roles, 0\% female in cyber security roles, 0 cyber security employees, cyber team: diverse employees. }
\end{singlespace}
}
\bigskip

Table~\ref{tab:debiased_lasso_results_additional} displays the svy LLasso estimates 
for the probability of experiencing a cyber security incident. 
Despite the large number of candidate predictors, cyber risk 
appears concentrated in a limited set of factors. Among firm 
characteristics, only large firm size 
($\widehat{\text{AME}}^{\text{DB}} = 0.057^{**}$) is statistically 
significant, while medium firms show no elevated risk relative to 
small firms. Several industry categories---Mining/Utilities/%
Construction, Manufacturing, Wholesale/Retail/Transport, and Other 
services---exhibit statistically significant positive coefficients 
relative to Professional services. Among digital practices, the 
only significant predictor is the use of confidential cloud 
solutions ($\widehat{\text{AME}}^{\text{DB}} = 0.067^{***}$); no 
other technology variable is statistically significant at 
conventional levels.

Two factors are associated with a lower likelihood of experiencing 
an incident: having over 60\% female employees in cyber security 
roles ($\widehat{\text{AME}}^{\text{DB}} = -0.056^{*}$) and 
holding cyber security insurance 
($\widehat{\text{AME}}^{\text{DB}} = -0.043^{*}$), though both 
are significant only at the 10\% level.

The sparsity of significant predictors is itself informative. The 
broad array of specific security measures included in the 
model -- anti-malware, email security, network security, access 
controls, and others -- are individually insignificant, suggesting 
that the presence or absence of any single protective technology 
does not meaningfully predict incidence at the firm level. Rather, 
the results point to firm scale and industry as the dominant 
correlates of cyber risk, with cloud storage as the only specific 
technology that is significantly associated with higher incidence. 
These patterns are consistent with the interpretation that larger 
firms and those in asset-intensive industries present larger attack 
surfaces and more valuable targets, while the protective 
associations of insurance and workforce diversity in cyber security 
roles warrant further investigation with designs capable of 
establishing causation.

\subsection{Interaction Effects} \label{sub: Interactions}

Using an adaptive Lasso approach with polynomial expansions, we investigate whether second-order (interaction) terms enhance the predictive accuracy of our \textit{Technological Efficiency} and \textit{Cyber Security Incidence} specifications. We estimate both a first-order (linear) model and a second-order model that includes all pairwise interactions among the explanatory variables (see \cite{Buhlmann-vandeGeer(2011)} for the theoretical background).

Table~\ref{tab:cv_polywog} reports mean-squared cross-validation (CV) errors under each specification, along with the penalty parameter $\lambda_{\mathrm{cv}}$.  For the \emph{Technological Efficiency} model, allowing second-order terms substantially lowers the CV error, suggesting that interactions play an important role in explaining efficiency gains from digital technologies. In the \emph{Cyber Security Incidence} model, the simpler first-order specification yields a slightly lower CV error, indicating that higher-order interactions do not improve predictions of cyber security incidence likelihood.\par 

A second-order polynomial specification is justified for the \emph{Technological Efficiency} model, therefore we re-estimate the \texttt{svy LLasso} model with interaction terms involving firm size, remote work, key technologies, and industry classifications. Table~\ref{tab:interaction} shows the significant interaction terms selected by the \texttt{svy LLasso} estimator, along with their debiased parameter estimates \(\tilde{\theta}^{DB}\) and \(p\)-values. Only the interaction coefficients that were selected by the Lasso and statistically significant at the 5\% level are included in Table \ref{tab:interaction}.\par

\begin{table}[ht!]
	\begin{center}
	\caption{Cross-Validation Results for Models With and Without Interaction Terms}
	\label{tab:cv_polywog}
		\begin{tabular}{l c c c c}
			\toprule
			\textbf{Model} & \textbf{Degree} & \(\lambda_{\mathrm{cv}}\) & \textbf{CV Error} & \textbf{Preferred Specification} \\
			\midrule
			\multirow{2}{*}{\textit{Technological Efficiency}} 
			& 1 & 0.00261 & 0.92442 &  \\
			& 2 & 0.00029 & 0.34142 & Second-order \\
			\midrule
			\multirow{2}{*}{\textit{Cyber Security Incidence}} 
			& 1 & 0.00685 & 0.87247 & First-order \\
			& 2 & 0.03132 & 0.86486 &  \\
			\bottomrule
		\end{tabular}
\end{center}
\footnotesize{\emph{Notes:} The table shows the mean-squared CV error from an adaptive Lasso specification with polynomial expansions of different degrees (1 = linear, 2 = second-order interactions). For each model, \(\lambda_{\mathrm{cv}}\) denotes the penalty parameter that minimizes the CV error. Based on these metrics, the second-order polynomial is preferred for the Technological Efficiency model, while a first-order specification is preferred for the Cyber Security Incidence model.}
\end{table}

{\small
	\setlength{\tabcolsep}{6pt}
	\begin{singlespace}
			\begin{longtable}{%
					l
					@{\extracolsep{0.01em}} D{.}{.}{3}
					@{\extracolsep{1.3em}} D{.}{.}{3}
					@{\extracolsep{1.3em}} D{.}{.}{3}
				}
				
				\caption{\texttt{svy LLasso} with Interactions: Technological Efficiency \label{tab:interaction}}\\
				
				\toprule
				\multicolumn{1}{l}{\textbf{Variables}}
				& \multicolumn{1}{c}{\textbf{Lasso}}
				& \multicolumn{1}{c}{$\tilde{\theta}^{DB}$}
				& \multicolumn{1}{c}{\textbf{p-value}}\\
				\midrule
				\endfirsthead

				\toprule
				\multicolumn{1}{l}{\textbf{Variables}}
				& \multicolumn{1}{c}{\textbf{Lasso}}
				& \multicolumn{1}{c}{$\tilde{\theta}^{DB}$}
				& \multicolumn{1}{c}{\textbf{p-value}}\\
				\midrule
				\endhead
				
				\bottomrule
				\endfoot
				
				\bottomrule
				\endlastfoot

				Intercept                                       & 1.095 & 3.953 & 0.012^{**}     \\
				Remote work                                         & 2.641 & 4.295 & 0.002^{***}    \\
				Online orders                                       & -2.008 & -4.555 & <0.001^{***}  \\
				IoT                                  & 2.249 & 4.873 & 0.009^{***}    \\
				Computer network                                    & -0.369 & -3.173 & 0.006^{***}    \\
				Fiber optic                                         & -1.197 & -3.019 & 0.009^{***}    \\
				Medium firm $\times$ Remote work                    & -0.884 & -1.428 & 0.040^{**}     \\
				Medium firm $\times$ IoT             & 1.901 & 5.554 & <0.001^{***}   \\
				Medium firm $\times$ Cyber security insurance       & -0.470 & -1.257 & 0.009^{***}    \\
				Medium firm $\times$ Website                        & 1.878 & 2.934 & 0.004^{***}     \\
				Medium firm $\times$ Company apps                   & -2.170 & -4.019 & 0.002^{***}    \\
				Medium firm $\times$ Manufacturing                  & -2.471 & -5.658 & <0.001^{***}  \\
				Medium firm $\times$ Other services                 & -2.677 & -8.226 & <0.001^{***}  \\
				Remote work $\times$ Female in ICT roles (1--20\%)  & 0.024 & -16.870 & 0.008^{***}    \\
				Remote work $\times$ CRM                            & 0.340 & 2.496 & 0.018^{**}      \\
				Remote work $\times$ Open source technologies       & -1.689 & -2.598 & 0.005^{***}    \\
				Remote work $\times$ Website                        & -1.864 & -3.230 & 0.003^{***}    \\
				Remote work $\times$ Social media                   & 2.361 & 3.824 & <0.001^{***}   \\
				ICT training $\times$ Foreign market                & 0.298 & 7.453 & 0.002^{***}     \\
				ICT training $\times$ CIM                           & -0.240 & -4.743 & 0.002^{***}    \\
				Female in ICT roles (1--20\%)  $\times$ Open source & 0.136 & 9.104 & 0.005^{***}     \\
				Large firm $\times$ CRM                             & -1.116 & -4.636 & 0.004^{***}    \\
				Online orders $\times$ Computer network             & 1.278 & 2.872 & 0.004^{***}     \\
				Online orders $\times$ ERP                          & -0.047 & -3.533 & 0.002^{***}    \\
				Online orders $\times$ Social media                 & 0.939 & 3.102 & 0.002^{***}     \\
				Online orders $\times$ Manufacturing                & 0.453 & 3.207 & 0.001^{***}    \\
				Foreign market $\times$ EDI                         & 1.704 & 5.321 & 0.003^{***}     \\
				Foreign market $\times$ CIM                         & -0.622 & -2.730 & 0.035^{**}     \\
				IoT $\times$ Computer network        & -0.890 & -3.443 & <0.001^{***}  \\
				IoT $\times$ CIM                     & 0.663 & 2.285 & 0.032^{**}      \\
				IoT $\times$ Manufacturing           & -2.497 & -4.684 & 0.002^{***}    \\
				Computer network $\times$ CRM                       & -2.041 & -2.983 & 0.006^{***}    \\
				Computer network $\times$ Advertising               & -0.191 & -3.003 & 0.018^{**}      \\
				Computer network $\times$ Fiber optic               & 1.525 & 3.576 & <0.001^{***}   \\
				Computer network $\times$ Other services            & -0.815 & -5.782 & 0.039^{**}     \\
				CRM $\times$ CIM                                    & -0.105 & -2.083 & 0.015^{**}      \\
				CRM $\times$ Advertising                            & -1.041 & -2.367 & 0.035^{**}      \\
				EDI $\times$ Advertising                            & -4.335 & -5.309 & <0.001^{***}  \\
				EDI $\times$ Online sales                           & -2.325 & -3.303 & 0.006^{***}     \\
				ERP $\times$ Manufacturing                           & 0.953 & 2.583 & 0.038^{**}      \\
				ERP $\times$ Wholesale/Retail/Transport             & 2.113 & 3.094 & 0.034^{**}      \\
				CIM $\times$ Mining/Utilities/Construction          & 1.256 & 6.701 & <0.001^{***}   \\
				CIM $\times$ Education/Health                       & -4.227 & -6.218 & <0.001^{***}  \\
				Advertising $\times$ Free advertising               & 2.001 & 2.886 & 0.021^{**}      \\
				Advertising $\times$ Website                        & 1.479 & 4.818 & 0.029^{**}      \\
				Advertising $\times$ Education/Health               & -5.083 & -9.202 & <0.001^{***}  \\
				Free advertising $\times$ Website                   & -0.161 & -6.123 & 0.013^{**}      \\
				Free advertising $\times$ Social media              & -0.298 & -3.753 & 0.010^{**}     \\
				Free advertising $\times$ Wholesale/Retail/Transport & -3.066 & -3.263 & 0.038^{**}    \\
				Free advertising $\times$ Education/Health          & 3.234 & 6.648 & <0.001^{***}   \\
				Website $\times$ Online sales                       & 1.572 & 6.424 & 0.011^{**}     \\
				Website $\times$ Mining/Utilities/Construction      & -1.338 & -4.640 & <0.001^{***}  \\
				Company apps $\times$ Social media                  & 5.323 & 10.627 & <0.001^{***}   \\
				Company apps $\times$ Manufacturing                 & -0.974 & -3.853 & 0.030^{**}     \\
				Company apps $\times$ Arts/Accommodation/Food       & -0.883 & -9.426 & 0.033^{**}     \\
				Social media $\times$ Wholesale/Retail/Transport    & -2.163 & -4.058 & <0.001^{***}  \\
				Fiber optic $\times$ Arts/Accommodation/Food        & 0.388 & 4.972 & 0.049^{**}      \\
				Fiber optic $\times$ Other services                 & 2.051 & 9.419 & <0.001^{***}   \\
				Online sales $\times$ Mining/Utilities/Construction & 2.554 & 8.200 & <0.001^{***}   \\
				Online sales $\times$ Manufacturing                 & -0.384 & -2.825 & 0.007^{***}    \\
				Online sales $\times$ Other services                & -1.933 & -9.947 & <0.001^{***}  \\
			\end{longtable}
\noindent \footnotesize{\emph{Notes}: Numeric values are rounded to three decimal places. $\tilde{\theta}^{DB}$ denotes the debiased logit Lasso coefficient estimate. Coefficients statistically significant at the 5\% level based on their $p$-values are reported. Significance levels: *** $p<0.01$, ** $p<0.05$, * $p<0.10$. Variable names are abbreviated: CRM (Customer Relationship Management), CIM (Client Information Management), EDI (Electronic Data Interchange), ERP (Enterprise Resource Planning). }
\end{singlespace}
}
\bigskip

Medium-sized firms exhibit positive and statistically significant interactions with the adoption of the IoT (\emph{Medium firm} $\times$ \emph{IoT}: $\tilde{\theta}^{DB}=5.554^{***}$) and firm website use (\emph{Medium firm} $\times$ \emph{Website}: $\tilde{\theta}^{DB}=2.934^{***}$). Remote work arrangements positively interact with social media (\emph{Remote work} $\times$ \emph{Social media}: $\tilde{\theta}^{DB}=3.824^{***}$) and customer relationship management software (\emph{Remote work} $\times$ \emph{CRM}: $\tilde{\theta}^{DB}=2.496^{**}$). \par 

The largest statistically significant positive interaction occurs between company-specific applications and social media use (\emph{Company apps} $\times$ \emph{Social media}: $\tilde{\theta}^{DB}=10.627^{***}$). The presence of female employees in ICT roles (1–20\%) interacts positively and significantly with open-source technology adoption (\emph{Female in ICT roles (1–20\%)} $\times$ \emph{Open source}: $\tilde{\theta}^{DB}=9.104^{***}$). Online sales and the Mining, Utilities, and Construction industry have a positive statistically significant interaction (\emph{Online sales} $\times$ \emph{Mining/Utilities/Construction}: $\tilde{\theta}^{DB}=8.200^{***}$). Additionally, free advertising positively interacts with firms in the Education and Health industry (\emph{Free advertising} $\times$ \emph{Education/Health}: $\tilde{\theta}^{DB}=6.648^{***}$). \par 

The interaction between EDI and advertising is negative and statistically significant (\emph{Electronic data interchange} $\times$ \emph{Advertising}: $\tilde{\theta}^{DB}=-5.309^{***}$). Similarly, medium-sized firms show a negative and statistically significant interaction with company-specific applications (\emph{Medium firm} $\times$ \emph{Company apps}: $\tilde{\theta}^{DB}=-4.019^{***}$). Remote work arrangements negatively interact with open-source technologies (\emph{Remote work} $\times$ \emph{Open source}: $\tilde{\theta}^{DB}=-2.598^{***}$). Online sales exhibit a negative interaction with firms in the Wholesale, Retail, and Transport industries (\emph{Online sales} $\times$ \emph{Wholesale/Retail/Transport}: $\tilde{\theta}^{DB}=-4.058^{***}$). \par

Medium-sized firms benefit from adopting IoT solutions and online platforms, likely due to greater resource availability compared to smaller firms. Remote work is positively associated with the efficiency of communication-oriented technologies, such as CRM and social media. The remote work variable itself also exhibits a strong, statistically significant positive effect on firm technological efficiency. Workforce diversity in ICT roles is positively correlated with technological efficiency, especially when adopting open-source systems. Industry-specific interactions produce varied effects depending on the technology and sector: online sales positively interact with mining, utilities, and construction, while advertising has a positive association in the education and health sectors. Conversely, interactions like online sales with wholesale, retail, and transport are negatively associated with firm technological efficiency.

\section{Conclusion}
\label{sec:conclusion}

This paper contributes new evidence on how Canadian businesses navigate the trade-off between digital technology adoption and heightened cyber security risk. Using data from Statistics Canada's SDTIU and CSCSC surveys, we construct a BDUS to gauge overall adoption levels and then evaluate how effectively businesses use these tools by modeling their technological efficiency. In addition, we employ a survey-weight-adjusted Lasso estimator and introduce a debiasing method for high-dimensional logit models to identify the predictors of technological efficiency and cyber security risk. 

The stochastic frontier analysis suggests that larger firms, remote work arrangements, and specific digital practices (e.g., open-source solutions, client information management systems)  are associated with proximity to the digital ``frontier.'' At the same time, a portion of businesses lag behind feasible adoption levels, as indicated by the high ratio of inefficiency to noise in the frontier estimations. \par 

Firms must balance efficiency gains against growing cyber vulnerabilities. Firms that adopt more sophisticated digital technologies or store sensitive data in the cloud often face elevated risk exposure. Our \texttt{svy LLasso} model on cyber incidence confirm that large firms and those using cloud-based services are more likely to report cyber security incidents. However, the predictive power of cyber risk does not improve with second-order interaction effects, suggesting that firm size and core technological choices are the primary predictors of cyber exposure. The two main variables associated with lower likelihood of cyber incidents were firms having cyber security insurance and firms that had a high representation of females in cyber security roles. \par

When it comes to technological efficiency simple linear relationships fail to capture the complexity of how organizational choices, workforce composition, and industry shape digital outcomes. By allowing second-order (interaction) terms, the  \texttt{svy LLasso} approach shows that certain combinations of variables such as \emph{medium firms} adopting IoT, or \emph{female ICT representation} interacting with advanced tools like AI can be particularly conducive to efficiency improvements. On the other hand, friction in implementing complex software like EDI or ERP can negate some of these benefits. \par 

The analysis demonstrates the importance of firm size and industry. While small firms are sometimes more ``locally efficient,” the resource advantages of larger organizations may facilitate deeper or more comprehensive integration of technologies. Industries also differ substantially. In resource- and asset-intensive sectors such as Mining or Construction, strong positive interactions emerge between targeted digital solutions and improved operational processes, whereas compliance-heavy fields like Education and Health exhibit more negative or complex relationships. \par 

Gender composition in ICT roles has meaningful associations with digital adoption outcomes. Although our results do not prove a causal mechanism, the recurring positive coefficients on interactions involving a share of female ICT staff and advanced technologies suggest that even partial gender diversity in technical teams is associated with higher returns to adopting new tools. This pattern is also seen in broader research suggesting that heterogeneity in skill sets and perspectives can catalyze creative problem-solving. \par

The results highlight the balance firms must strike between achieving efficiency gains from digital technologies and managing increased cyber risks. Larger, digitally advanced firms approach their efficiency frontier yet face elevated cyber vulnerability, especially when security practices fall short. Remote work arrangements are associated with both higher digital adoption and technological efficiency, without increasing cyber risk exposure. Female representation in ICT and cyber security roles is consistently associated with better outcomes across adoption, efficiency, and cyber security. Certain cyber security practices, particularly obtaining cyber security insurance and ensuring gender diversity within cyber security roles significantly reduce incident likelihood. Reliance on cloud-based services, notably confidential cloud storage, emerges as a risk factor associated with higher cyber vulnerability. These results suggest policymakers should implement targeted digital strategies tailored by industry and firm size to boost technological efficiency, while simultaneously establishing baseline cyber security practices such as insurance coverage and workforce diversity to mitigate cyber threats effectively.

\subsubsection*{Limitations}
Several limitations warrant acknowledgment. Our cross-sectional design captures associations at a single point in time and cannot establish causal relationships or trace the dynamics of adoption and cyber security investment. The relationships we document may reflect reverse causality. For 
example, firms that experience a cyber security incident may 
subsequently invest in more sophisticated digital tools or adopt 
additional security measures, generating a positive association 
between incidence and the practices we treat as predictors. Our 
cross-sectional design cannot distinguish these temporal sequences. 
The results may also be confounded by omitted firm-level 
heterogeneity, such as managerial sophistication or risk 
preferences, which likely affect both digital adoption decisions 
and the quality of cyber security practices but are unobserved in 
our data. 

Additionally, survey-based measures of digital adoption and cyber incidents may be subject to reporting biases, as firms might underreport security breaches due to reputational concerns or overreport technology adoption due to social desirability.

\newpage
\bibliographystyle{myagsm}
\bibliography{SDTIU}

\newpage
\appendix

\section{Technical Appendix}\label{sec: tech app}
\subsection{Additional Details on the Implementation of the \texttt{svy LLasso} Procedure}\label{subsec: app1}
\paragraph*{Tuning Parameter Selection.} In both the empirical analysis and the small Monte Carlo simulations below, the logit Lasso procedure \texttt{svy LLasso} is fitted using the \texttt{R} package \texttt{glmnet}. For the tuning parameter \( \lambda \), we use the package’s default value, which is selected via 10-fold cross-validation using the loss function \texttt{auc} (area under the ROC curve).

\paragraph*{Average Marginal Effect in the Logit Model.} 
We report the marginal effects (ME) for each variable along with the coefficient estimates for logit models in the tables. Let $\Lambda(z) = \exp(z)/(1 + \exp(z))$ be the logistic distribution function. For a dummy regressor $\tilde{x}_{ij}$, where $j = 1, \dots, p$ and $i = 1, \dots, n$, the ME is defined as:
\[
\text{ME}_{ij}(\theta)\equiv\Lambda({x}_{i}'\theta)\vert_{\tilde{x}_{ij}=1}-\Lambda({x}_{i}'\theta)\vert_{\tilde{x}_{ij}=0}.
\]
Given the survey weights $\{w_i\}_{i=1}^n$ corresponding to the observations $\{(y_i, x_i')'\}_{i=1}^n$, the AME of the $j$-th regressor is given by:
\[
\text{AME}_j=\text{AME}_j(\theta_0)\equiv \E\left[\frac{1}{\sum_{i=1}^n w_i} \sum_{i=1}^n w_i \text{ME}_{ij}(\theta_0)\right],
\]
where $\theta_0$ is the true parameter value, and the expectation is with respect to the regressors' distribution. An estimator for $\text{AME}_j$ is:
\[
\widehat{\text{AME}}_j(\hat{\theta}) \equiv \frac{1}{\sum_{i=1}^n w_i} \sum_{i=1}^n w_i \left(\Lambda({x}_{i}'\hat{\theta})\vert_{\tilde{x}_{ij}=1} - \Lambda({x}_{i}'\hat{\theta})\vert_{\tilde{x}_{ij}=0}\right),
\]
where $\hat{\theta} = (\hat{\alpha}, \hat{\beta}')'$ is the estimated parameter vector, e.g., the \texttt{svy LLasso} estimator. The debiased logit Lasso estimator of $\text{AME}_j$ is then constructed using the one-step iteration provided in the equation \eqref{eq: rho DB}.

\subsection{Post-Selection Inference for Survey-GLM}\label{subsec: app2}
Consider the density of a scalar outcome variable $y_{i}$ given a $(p+1)\times 1$ vector of covariates $x_{i}$ (which includes a constant) specified as 
\begin{equation*}
	f(y_i\vert x_i,\theta_0)=\exp(y_ix_i'\theta_0-a(x_i'\theta_0))c(y_i),\quad i=1,\dots, n,
\end{equation*}
where $\theta_0$ is the true value of the parameter vector 
$\theta\in\mathbb{R}^{p+1}$, and $a(\cdot)$ and $c(\cdot)$ are known functions. 
The combined SDTIU and CSCSC data were collected using a stratified sampling scheme, wherein units within each stratum are sampled independently with equal probability. In line with this, we treat $w_i$ as fixed \citep[see][]{Wooldridge2001}, and assume $\{(y_i, x_i')'\}_{i=1}^n$ to be independent.\par
Let $g(y, x'\theta)\equiv -\log f(y, x'\theta)$ and define the weighted log-likelihood function as:
\begin{equation}\label{eq: WLL}
	L_n(\theta)\equiv -n^{-1}\sum_{i=1}^nw_i g(y_{i}, x_{i}'\theta).
\end{equation}
The score function, the sample information and negative Hessian matrices corresponding to \eqref{eq: WLL} are defined as 
\begin{align}
	S(\theta)&\equiv \frac{\partial L_n(\theta)}{\partial \theta}
	=-n^{-1}\sum_{i=1}^nw_ix_{i}\dot{g}(y_i,x_i'\theta),\quad 
	\dot{g}(y, t) \equiv \frac{\partial g(y, t)}{\partial t},\label{eq: scoref}\\
	\hat{I}(\theta)
	&\equiv n^{-1}\sum_{i=1}^nw_i^2x_ix_i'\dot{g}(y_i,x_i'\theta)^2,\label{eq: info}\\
	\hat{H}(\theta)
	&\equiv -\frac{\partial^2 {L}_n(\theta)}{\partial \theta\partial \theta'}=n^{-1}\sum_{i=1}^nw_ix_ix_i'\ddot{g}(y_i,x_i'\theta),\quad \ddot{g}(y, t)\equiv \frac{\partial^2 g(y, t)}{\partial t^2}.\label{eq: hessian}
\end{align}
Moreover, we define $H(\theta_0)\equiv \E[\hat{H}(\theta_0)]$ and $I(\theta_0)\equiv \E[\hat{I}(\theta_0)]$.\par 

We will using the following notations in the assumptions and the proof of Proposition \ref{prop: DB} below. 
Let $\lambda_{\min}(A)$ and $\lambda_{\max}(A)$ denote the smallest and the largest eigenvalue of a symmetric matrix $A$, respectively. For a real matrix 
$A=(a_{ij})$, let $\Vert A\Vert_{\infty}\equiv \max_{i,j}\left|a_{ij}\right|$, and $\Vert A\Vert=\sqrt{\mathrm{tr}(A'A)}$ 
and $\Vert A\Vert_2=\sqrt{\lambda_{\max}(A'A)}$ denote its Frobenius and spectral norms, respectively. 
The sub-Gaussian norm of a random variable $X$ is defined as 
$\Vert X\Vert_{\psi_2}\equiv \sup_{m\geq 1}m^{-1/2}(\E[\vert X\vert^m])^{1/m}.$ The sub-Gaussian norm for the random vector is defined as 
$\Vert X\Vert_{\psi_2}\equiv \sup_{\Vert b\Vert=1}\Vert X'b\Vert_{\psi_2}$.\par 
We establish the asymptotic validity of the debiasing method under the following 
assumptions imposed directly on the negative log-density function $g(y, t)$ which are similar to the assumptions employed in \cite{vandeGeer-etal(2014)} and 
\cite{Xia2023}. Let $X=[x_1,\dots, x_n]'$.  
\begin{assumption}[Asymptotic validity]\label{A: AsyValid}
	\leavevmode
	\begin{enumerate}[label={(\alph*)}]
		\item \label{AsyValid max}
		$\{(y_i, x_i')'\}_{i=1}^n$ are independent with $\max_{1\leq i\leq n}a_i<C_u<\infty$ a.s. where  
		\begin{equation*}		
			a_i\in\{\Vert x_i\Vert_{\psi_2}, \Vert x_{i}\Vert_{\infty}, \Vert X\theta_0\Vert_{\infty}\}.
		\end{equation*}		 
		Moreover, $w_i$ is non-random with $0<C_l<w_i<C_u$ for all $n, i$.
		\item \label{AsyValid eval}
		For $A\in \{H(\theta_0), I(\theta_0), \E[n^{-1}X'X]\}$, there exist positive constants $\lambda_l$ and $\lambda_u$ such that 
		$0<\lambda_l\leq \lambda_{\min}(A)\leq \lambda_{\max}(A)\leq \lambda_{u}<\infty$. 
		\item \label{AsyValid Lip} 
		The function $g(y, t)\equiv a(t)-yt-\log c(y)$ is convex in $t\in\mathbb{R}$ for all $y$, and 
		twice differentiable with respect to $t$ for all $(y,t)$. 
		There exist a positive definite matrix $H$ and $\eta>0$ such that $\lambda_{\min}(H)>\lambda_l>0$ and 
		\begin{equation}\label{A: Strong Con}
			n^{-1}\sum_{i=1}^n\E[w_i(g(y_i,x_i'\theta)-g(y_i,x_i'\theta_0))]\geq \Vert H^{1/2}(\theta-\theta_0)\Vert^2
		\end{equation}
		for all $\Vert X(\theta-\theta_0)\Vert_{\infty}<\eta$. Furthermore, $\ddot{g}(y,t)$ is Lipschitz with some constant $L_0>0$:
		\begin{equation}\label{A: dg lip}
			\max_{t_0 \in \{x_i' \theta_0\}}   \sup_{\max(|t-t_0|, |\tilde{t} - t_0|) \leq \eta}  \sup_{y \in \mathcal{Y}}\frac{| \ddot{g}(y,t)-\ddot{g}(y,\tilde{t})|}{|t-\tilde{t}|}\leq L_0,
		\end{equation}
		and 
		\begin{align}
			&\max_{t_0 \in \{x_i' \theta_0\}}\sup_{y \in \mathcal{Y}} |\dot{g}(y,t_0)| \leq C_u,\label{A: dg1}\\
			&\max_{t_0 \in \{x_i' \theta_0\}} \sup_{|t-t_0|\leq \eta} \sup_{y \in \mathcal{Y}} |\ddot{g}(y,t)| \leq C_u.\label{A: dg2} 
		\end{align}
	\end{enumerate}
\end{assumption}
For discussions of these assumptions, we refer to \cite{JT2023}. 

\paragraph*{Simulations.} We conduct a simulation experiment to verify the robustness of the debiased logit Lasso inference. We first generate $N=10,000$ draws from a standard logit model as follows:
\begin{equation}
	y_i \sim \mathrm{Bernoulli}(\pi_i),
\end{equation}
where $\theta_0 = (1, 1, 1, 0_{1\times (p-2)})'$, $\tilde{x}_{ij} \sim \text{i.i.d.\ } \mathrm{Bernoulli}(0.5)$ for $j=1,\dots,p$ and $i=1,\dots,N$, $x_i = (1, \tilde{x}_i')'$, and $\pi_i = x_i'\theta_0$.

The population is then stratified into four strata of sizes 1,000, 2,000, 3,000, and 4,000. From each stratum, we draw 50 and 100 observations with replacement, yielding stratified samples of size $n=200$ and $n=400$, respectively. Observation weights are $w_i = 0.1, 0.2, 0.3, 0.4$, corresponding to the four strata. To evaluate the impact of regressor dimensionality, we set $p$ such that $\frac{p}{n} \in \{0.01, 0.025, 0.05, 0.1, 0.25, 0.5\}$ for each $n \in \{200, 400\}$. The true AME for $\theta_{(2)} = \beta_1$ is 0.11.\par 
We assess the empirical size of the tests by separately testing two null hypotheses:
\begin{equation}\label{eq: H0 sim}
	H_0: \theta_{(2)} = 1, \quad H_0: \mathrm{AME}_2 = 0.11.
\end{equation}
The empirical sizes of the DB test and the standard survey \( t \)-test (\( t_{\mathrm{svy}} \)) at the 5\% nominal level are presented in the table below. The standard survey logit \( t_{\mathrm{svy}} \) test overrejects by a wide margin, while the DB test exhibits reasonably accurate null rejection rates for both hypotheses in most cases, confirming its robustness to regressor dimensionality.
\begin{table}[htbp]
	\caption{Empirical rejection frequencies of the tests for $H_0: \theta_{(2)}=1$ and $H_0:\mathrm{AME}_2=0.11$ at $5\%$ level. Standard stratified sampling.}
	\label{tab:Level0}%
	\begin{center}
		\begin{tabular}{lcccccc}
			\toprule
			\multicolumn{1}{l}{Tests} &	\multicolumn{1}{c}{$p=2$} &\multicolumn{1}{c}{$p=5$} & \multicolumn{1}{c}{$p=10$} & \multicolumn{1}{c}{$p=20$} & 	\multicolumn{1}{c}{$p=50$}& 	\multicolumn{1}{c}{$p=100$}\\
			\midrule
			\multicolumn{7}{c}{$H_0:\theta_{(2)}=1$, $n=200$}\\
			\midrule
			DB & 5.0 & 4.4 & 3.7 & 3.1 & 4.5 & 3.3 \\ 			
			$t_{\mathrm{svy}}$ & 6.2 & 6.4 & 8.0 & 8.7 & 36.0 & 94.9 \\ 		
			\midrule
			\multicolumn{7}{c}{$H_0:\mathrm{AME}_{2}=0.11$, $n=200$}\\
			\midrule 
			DB & 5.4 & 5.3 & 4.6 & 3.7 & 3.5 & 1.4 \\ 				
			$t_{\mathrm{svy}}$ & 5.7 & 7.7 & 7.4 & 8.2 & 50.9 & 93.3 \\ 		
			\bottomrule
		\end{tabular}
	\end{center}
	
	\begin{center}
		\begin{tabular}{lcccccc}
			\toprule
			\multicolumn{1}{l}{Tests} &	\multicolumn{1}{c}{$p=4$} & \multicolumn{1}{c}{$p=10$} & \multicolumn{1}{c}{$p=20$}  & \multicolumn{1}{c}{$p=40$} & 	\multicolumn{1}{c}{$p=100$} & \multicolumn{1}{c}{$p=200$}\\
			\midrule
			\multicolumn{7}{c}{$H_0:\theta_{(2)}=1$, $n=400$}\\
			\midrule 
			DB & 4.8 & 4.4 & 6.0 & 3.7 & 5.6 & 3.9 \\ 
			$t_{\mathrm{svy}}$ & 5.0 & 5.1 & 6.3 & 15.9 & 40.4 & 98.3 \\ 
			\midrule
			\multicolumn{7}{c}{$H_0:\mathrm{AME}_{2}=0.11$, $n=400$}\\
			\midrule 
			DB & 4.5 & 4.9 & 5.8 & 5.0 & 4.6 & 3.3 \\ 
			$t_{\mathrm{svy}}$ & 5.3 & 6.9 & 9.1 & 10.8 & 46.8 & 93.7 \\ 
			\bottomrule
		\end{tabular}
	\end{center}
	\footnotesize{Notes: $n=200, 400$. DB and $t_{\mathrm{svy}}$ denote 
		the debiased Lasso and standard survey-weighted $t$ tests respectively. 1000 simulation replications.}
\end{table}

\subsection{Proof of Proposition \ref{prop: DB}}\label{subsec: proof DB}
We first prove the following lemma, which establishes the asymptotic distribution of a studentized quantity involving the expected Hessian and information matrices, as well as the score function, all evaluated at the true parameters.
\begin{lemma}\label{lem: AD1}
	Let Assumption \ref{A: AsyValid} hold and $p^{1+\delta_0}/n\to 0$ for some $0<\delta_0\leq 1$. 
	Then, as $n\to\infty$
	\begin{equation*}
		{\left(\dot{\rho}(\theta_0)'H(\theta_0)^{-1}I(\theta_0)H(\theta_0)^{-1}\dot{\rho}(\theta_0)\right)^{-1/2}}\dot{\rho}(\theta_0)'H(\theta_0)^{-1}n^{1/2}S(\theta_0)
		\cond \Norm{0,I_r}.
	\end{equation*}
\end{lemma}
\begin{proof}[Proof of Lemma \ref{lem: AD1}]
	Let $s_i(\theta_0)\equiv w_ix_i\dot{g}(y_i, x_i'\theta_0)$, $X_{ni}\equiv n^{-1/2}\dot{\rho}(\theta_0)'H(\theta_0)^{-1}s_i(\theta_0)$
	and $\Sigma_n\equiv \V[\sum_{i=1}^nX_{ni}]=\dot{\rho}(\theta_0)'{H}(\theta_0)^{-1}
	{I}(\theta_0){H}(\theta_0)^{-1}\dot{\rho}(\theta_0)$. Let $\nu_n\equiv \lambda_{\min}(\Sigma_n)$.  
	We will verify the conditions of the multivariate Lindeberg-Feller CLT (see e.g. Theorem 9.3 of \cite{Hansen(2022)}). 
	First note that $\E[X_{ni}]=0$ because $\E[s_i(\theta_0)\vert x_i]=-\E[x_iw_i(y_i-\dot{a}(x_i'\theta_0))\vert x_i]=0$. 
	Moreover, we have  
	\begin{align*}
		\nu_n&=\min_{\tau\in\mathbb{R}^r\setminus \{0\}}\frac{\tau'\dot{\rho}(\theta_0)'{H}(\theta_0)^{-1}
			{I}(\theta_0){H}(\theta_0)^{-1}\dot{\rho}(\theta_0)\tau}
		{\tau'\tau}\notag\\
		&\geq \min_{\tau\in\mathbb{R}^r\setminus \{0\}}\frac{\tau'\dot{\rho}(\theta_0)'{H}(\theta_0)^{-1}
			{I}(\theta_0){H}(\theta_0)^{-1}\dot{\rho}(\theta_0)\tau}
		{\tau'\dot{\rho}(\theta_0)'\dot{\rho}(\theta_0)\tau}\min_{\tau\in\mathbb{R}^r\setminus \{0\}}\frac
		{\tau'\dot{\rho}(\theta_0)'\dot{\rho}(\theta_0)\tau}{\tau'\tau}\\
		&\geq \lambda_{\min}({H}(\theta_0)^{-1}
		{I}(\theta_0){H}(\theta_0)^{-1})\lambda_{\min}(\dot{\rho}(\theta_0)'\dot{\rho}(\theta_0))\\
		&\geq \lambda_{\min}({H}(\theta_0)^{-1})
		\lambda_{\min}({I}(\theta_0))\lambda_{\min}({H}(\theta_0)^{-1})\lambda_{\min}(\dot{\rho}(\theta_0)'\dot{\rho}(\theta_0))\\
		&=\frac{\lambda_{\min}({I}(\theta_0))}{(\lambda_{\max}({H}(\theta_0))^2}\lambda_{\min}(\dot{\rho}(\theta_0)'\dot{\rho}(\theta_0))\\
		&\geq \lambda_l^2/\lambda_u^2.
	\end{align*}
	where the first inequality follows from the extremal property of $\lambda_{\min}(\cdot)$, the second inequality is the eigenvalue product inequality (\cite{Hansen(2022b)}) and the last inequality is by  
	Assumption \ref{A: AsyValid}\ref{AsyValid eval}. Next, 
	we will verify the Lindeberg condition: for $\delta=\frac{2}{\delta_0}>0$ and any $\epsilon>0$
	\begin{align}\label{eq: Lyap}
		\frac{1}{\nu_n^{2}}\sum_{i=1}^n\E[\Vert X_{ni}\Vert^{2}1(\Vert X_{ni}\Vert\geq (\epsilon\nu_n^2)^{1/2})]	\leq \frac{1}{\nu_n^{2+\delta}\epsilon^{\delta/2}}\sum_{i=1}^n\E[\Vert X_{ni}\Vert^{2+\delta}]\to 0.
	\end{align}
	First, note that 
	\begin{align}
		\Vert \dot{\rho}(\theta_0)'{H}(\theta_0)^{-1}x_i\Vert ^{2+\delta}
		&\leq \Vert\dot{\rho}(\theta_0)\Vert^{2+\delta}\left(\Vert {H}(\theta_0)^{-1}x_i\Vert^2\right)^{1+\delta/2}\notag\\
		&\leq r^{1+\delta/2}\Vert\dot{\rho}(\theta_0)\Vert_2^{2+\delta}\left(\lambda_{\max}({H}(\theta_0)^{-1}{H}(\theta_0)^{-1})\Vert x_i\Vert^2\right)^{1+\delta/2}\notag\\
		&\leq r^{1+\delta/2}\lambda_u^{2+\delta}\left(\frac{\Vert x_i\Vert^2}{(\lambda_{\min}({H}(\theta_0)))^2}\right)^{1+\delta/2}\notag\\
		&\leq r^{1+\delta/2}\lambda_u^{2+\delta}\frac{(p+1)^{1+\delta/2}C_u^{2+\delta}}{\lambda_l^{2+\delta}}.\label{eq: Lyap2}
	\end{align}
	where the first inequality is by Cauchy-Schwarz, the second inequality is by the inequality 
	$\Vert\dot{\rho}(\theta_0)\Vert\leq r^{1/2}\Vert\dot{\rho}(\theta_0)\Vert_2$ and 
	the extremal property of $\lambda_{\max}(\cdot)$, the third inequality is by the eigenvalue product inequality (\cite{Hansen(2022b)}, Appendix B), and the last inequality is by Assumption \ref{A: AsyValid}\ref{AsyValid max} and \ref{AsyValid eval}.
	Thus, using $\vert w_i\vert^{2+\delta}\vert \dot{g}(y_i, x_i'\theta_0)\vert^{2+\delta}
	\leq C_u^{4+2\delta}$ and \eqref{eq: Lyap2}, we have 
	\begin{align}
		\sum_{i=1}^n\Vert X_{ni}\Vert^{2+\delta}
		&\leq \frac{1}{n^{1+\delta/2}}\sum_{i=1}^n
		\Vert \dot{\rho}(\theta_0)'{H}(\theta_0)^{-1}x_i\Vert ^{2+\delta}
		\vert w_i\vert^{2+\delta}\vert \dot{g}(y_i, x_i'\theta_0)\vert^{2+\delta}\notag\\
		&\leq \frac{1}{n^{\delta/2}}r^{1+\delta/2}\lambda_u^{2+\delta}\frac{(p+1)^{1+\delta/2}C_u^{2+\delta}}{\lambda_l^{2+\delta}}C_u^{4+2\delta}\notag\\
		&\leq \left(\frac{(p+1)^{1+\delta_0}}{n}\right)^{1/\delta_0}r^{1+\delta/2}\lambda_u^{2+\delta}\frac{C_u^{6+3\delta}}{\lambda_l^{2+\delta}}\notag\\
		&\to 0.\notag
	\end{align}
	This verifies \eqref{eq: Lyap} and the result follows. 
\end{proof}

\paragraph*{Proof of Proposition \ref{prop: DB}}
By the mean value expansion, 
\begin{equation}\label{eq: mve}
	S(\theta_0)=S(\hat{\theta})+\hat{H}({\theta}^{*})(\hat{\theta}-\theta_0)=S(\hat{\theta})+\hat{H}(\hat{\theta})(\hat{\theta}-\theta_0)+R,
\end{equation}
where ${\theta}^{*}$ is the mean-value between $\hat{\theta}$ and $\theta_0$, and $R=[R_1,\dots, R_{p+1}]'$ with 
\begin{equation}\label{eq: rem0}
	R_j\equiv n^{-1}\sum_{i=1}^n(\ddot{g}(y_i, x_i'\theta^{*})-\ddot{g}(y_i, x_i'\hat{\theta}))
	w_ix_{ij}x_i'(\theta_0-\hat{\theta}).
\end{equation}
Note that since $\dot{\rho}(\theta)$ is locally Lipschitz in a neighborhood of $\theta_0$, 
with probability approaching 1
$\Vert\dot{\rho}(\bar{\theta})-\dot{\rho}(\hat{\theta})\Vert \leq B_0\Vert \bar{\theta}-\hat{\theta}\Vert$ for some $B_0=O(1)$. 
Also, since 
\begin{equation}\label{eq: DB rho mve}
	{n}^{1/2}(\rho(\hat{\theta})-\rho({\theta}_0))=	
	\dot{\rho}(\bar{\theta})'{n}^{1/2}(\hat{\theta}-\theta_0),
\end{equation} 
where $\bar{\theta}$ is a mean-value between $\hat{\theta}$ 
and $\theta_0$, we have 
\begin{align}
	n^{1/2}\Vert \dot{\rho}({\hat{\theta}})-\dot{\rho}(\bar{\theta})\Vert \Vert\hat{\theta}-\theta_0\Vert
	&=n^{1/2}B_0\Vert\hat{\theta}-\bar{\theta}\Vert \Vert\hat{\theta}-\theta_0\Vert=O_p(n^{1/2}m_0\lambda^2)\notag\\
	&=o_p(1),\label{eq: DB rho2}
\end{align} 
where the last line is by $n^{1/2}m_0\lambda^2= n^{-1/2}m_0C^2\log p\leq C^2m_0 (p/n)^{1/2}\log p=o(1)$. Then, 
\begin{align*}
	&n^{1/2}(\tilde{\rho}-\rho(\theta_0))\\
	&=
	n^{1/2}(\rho(\hat{\theta})-\rho(\theta_0))
	+\dot{\rho}(\hat{\theta})'\hat{H}(\hat{\theta})^{-1}n^{1/2}S(\hat{\theta})\\
	&=
	n^{1/2}\dot{\rho}(\bar{\theta})'(\hat{\theta}-\theta_0)
	+n^{1/2}\dot{\rho}(\hat{\theta})'\hat{H}(\hat{\theta})^{-1}S(\theta_0)
	-n^{1/2}\dot{\rho}(\hat{\theta})'(\hat{\theta}-\theta_0)
	-n^{1/2}\dot{\rho}(\hat{\theta})'\hat{H}(\hat{\theta})^{-1}R\\
	&=n^{1/2}\dot{\rho}(\hat{\theta})'\hat{H}(\hat{\theta})^{-1}S(\theta_0)
	-n^{1/2}\dot{\rho}(\hat{\theta})'\hat{H}(\hat{\theta})^{-1}R+o_p(1),
\end{align*}
where the first equality is by the definition of $\tilde{\rho}$, the second equality is by \eqref{eq: mve} and \eqref{eq: DB rho mve},
and the third is by \eqref{eq: DB rho2}. Below, the proof will be completed in three steps: the first two steps establish
\begin{align}
	&\dot{\rho}(\hat{\theta})'\hat{H}(\hat{\theta})^{-1}\hat{I}(\hat{\theta})\hat{H}(\hat{\theta})^{-1}
	\dot{\rho}(\hat{\theta})-\dot{\rho}({\theta}_0)'{H}({\theta}_0)^{-1}{I}({\theta}_0){H}({\theta}_0)^{-1}
	\dot{\rho}({\theta}_0)=o_p(1),\notag\\
	&n^{1/2}\dot{\rho}(\hat{\theta})'\hat{H}(\hat{\theta})^{-1}S(\theta_0)-n^{1/2}\dot{\rho}({\theta}_0)'{H}({\theta}_0)^{-1}S(\theta_0)=o_p(1),
\end{align}
and the third step verifies $n^{1/2}\dot{\rho}(\hat{\theta})'\hat{H}(\hat{\theta})^{-1}R=o_p(1)$. 
It will then follow that 
\begin{align*}
	&\left[\dot{\rho}(\hat{\theta})'\hat{H}(\hat{\theta})^{-1}\hat{I}(\hat{\theta})\hat{H}(\hat{\theta})^{-1}
	\dot{\rho}(\hat{\theta})\right]^{-1/2}n^{1/2}(\tilde{\rho}-\rho(\theta_0))\\
	&=\left[\dot{\rho}(\hat{\theta})'\hat{H}(\hat{\theta})^{-1}\hat{I}(\hat{\theta})\hat{H}(\hat{\theta})^{-1}\dot{\rho}(\hat{\theta})\right]^{-1/2}
	\left[n^{1/2}\dot{\rho}(\hat{\theta})'\hat{H}(\hat{\theta})^{-1}S(\theta_0)-n^{1/2}\dot{\rho}(\hat{\theta})'\hat{H}(\hat{\theta})^{-1}R+o_p(1)\right]
	\\
	&=\left[\dot{\rho}(\theta_0)'{H}({\theta}_0)^{-1}{I}({\theta}_0){H}({\theta}_0)^{-1}\dot{\rho}(\theta_0)\right]^{-1/2}n^{1/2}\dot{\rho}(\theta_0)'{H}({\theta}_0)^{-1}S(\theta_0)
	+o_p(1).
\end{align*}
Finally, applying Lemma \ref{lem: AD1} of and Slutsky's lemma give the desired result. 
\paragraph*{Step 1: $\dot{\rho}(\hat{\theta})'\hat{H}(\hat{\theta})^{-1}\hat{I}(\hat{\theta})\hat{H}(\hat{\theta})^{-1}
	\dot{\rho}(\hat{\theta})-\dot{\rho}({\theta}_0)'{H}({\theta}_0)^{-1}{I}({\theta}_0){H}({\theta}_0)^{-1}
	\dot{\rho}({\theta}_0)=o_p(1)$.}~\\
First, by the triangle inequality
\begin{align}
	&\Vert \dot{\rho}(\hat{\theta})'\hat{H}(\hat{\theta})^{-1}\hat{I}(\hat{\theta})\hat{H}(\hat{\theta})^{-1}
	\dot{\rho}(\hat{\theta})-\dot{\rho}({\theta}_0)'{H}({\theta}_0)^{-1}{I}({\theta}_0){H}({\theta}_0)^{-1}
	\dot{\rho}({\theta}_0)\Vert_2\notag\\
	&\leq \Vert \dot{\rho}(\hat{\theta})'\left[\hat{H}(\hat{\theta})^{-1}\hat{I}(\hat{\theta})\hat{H}(\hat{\theta})^{-1}
	-{H}({\theta}_0)^{-1}{I}({\theta}_0){H}({\theta}_0)^{-1}\right]\dot{\rho}(\hat{\theta})\Vert_2\notag\\
	&\quad+\Vert \dot{\rho}(\hat{\theta})'{H}({\theta}_0)^{-1}{I}({\theta}_0){H}({\theta}_0)^{-1}(\dot{\rho}(\hat{\theta})-\dot{\rho}(\theta_0))\Vert_2\notag\\
	&\quad+\Vert (\dot{\rho}(\hat{\theta})-\dot{\rho}(\theta_0))'{H}({\theta}_0)^{-1}{I}({\theta}_0){H}({\theta}_0)^{-1}\dot{\rho}(\theta_0)\Vert_2.\label{eq: db tr ineq}
\end{align}
Consider the first term on the right-hand side of \eqref{eq: db tr ineq}. By Cauchy-Schwarz inequality, 
\begin{align}
	&\Vert \dot{\rho}(\hat{\theta})'\left[\hat{H}(\hat{\theta})^{-1}\hat{I}(\hat{\theta})\hat{H}(\hat{\theta})^{-1}
	-{H}({\theta}_0)^{-1}{I}({\theta}_0){H}({\theta}_0)^{-1}\right]\dot{\rho}(\hat{\theta})\Vert_2\notag\\
	&\leq \Vert\hat{H}(\hat{\theta})^{-1}\hat{I}(\hat{\theta})\hat{H}(\hat{\theta})^{-1}
	-{H}({\theta}_0)^{-1}{I}({\theta}_0){H}({\theta}_0)^{-1}\Vert_2 \Vert\dot{\rho}(\hat{\theta})\Vert_2^2,\label{eq: HIH negl}
\end{align}
After rearranging and using the triangle and Cauchy-Schwarz inequalities
\begin{align}
	&\Vert \hat{H}(\hat{\theta})^{-1}\hat{I}(\hat{\theta})\hat{H}(\hat{\theta})^{-1}
	-{H}({\theta}_0)^{-1}{I}({\theta}_0){H}({\theta}_0)^{-1}\Vert_2\notag\\
	&=\Vert(\hat{H}(\hat{\theta})^{-1}-{H}({\theta}_0)^{-1})\hat{I}(\hat{\theta})\hat{H}(\hat{\theta})^{-1}
	+{H}({\theta}_0)^{-1}(\hat{I}(\hat{\theta})\hat{H}(\hat{\theta})^{-1}-{I}({\theta}_0){H}({\theta}_0)^{-1})\Vert_2,\notag\\
	&\leq \Vert\hat{H}(\hat{\theta})^{-1}-{H}({\theta}_0)^{-1}\Vert_2\Vert\hat{I}(\hat{\theta})\Vert_2\Vert\hat{H}(\hat{\theta})^{-1}\Vert_2
	+\Vert{H}({\theta}_0)^{-1}\Vert_2\Vert \hat{I}(\hat{\theta})\hat{H}(\hat{\theta})^{-1}-{I}({\theta}_0){H}({\theta}_0)^{-1}\Vert_2.\label{eq: HI diff1}
\end{align}
For the first summand of \eqref{eq: HI diff1}, by Lemma A.2 of \cite{JT2023}
\begin{equation}
	\Vert\hat{H}(\hat{\theta})^{-1}-{H}({\theta}_0)^{-1}\Vert_2\Vert\hat{I}(\hat{\theta})\Vert_2\Vert\hat{H}(\hat{\theta})^{-1}\Vert_2=o_p(1).
\end{equation}
For the second factor in the second summand of \eqref{eq: HI diff1}, using the triangle and Cauchy-Schwarz inequalities 
\begin{align}
	&\Vert\hat{I}(\hat{\theta})\hat{H}(\hat{\theta})^{-1}-I(\theta_0){H}({\theta}_0)^{-1}\Vert_2\notag\\
	&=\Vert (\hat{I}(\hat{\theta})-{I}({\theta}_0))(\hat{H}(\hat{\theta})^{-1}-{H}({\theta}_0)^{-1})
	+(\hat{I}(\theta_0)-I(\theta_0)){H}({\theta}_0)^{-1}
	+I(\theta_0)(\hat{H}(\theta_0)^{-1}-H(\theta_0)^{-1})\Vert_2\notag\\
	&\leq \Vert\hat{I}(\hat{\theta})-{I}({\theta}_0)\Vert_2\Vert\hat{H}(\hat{\theta})^{-1}-{H}({\theta}_0)^{-1}\Vert_2
	+\Vert\hat{I}(\theta_0)-I(\theta_0)\Vert_2\Vert{H}({\theta}_0)^{-1}\Vert_2\notag\\
	&\quad+\Vert I(\theta_0)\Vert_2\Vert\hat{H}(\theta_0)^{-1}-H(\theta_0)^{-1}\Vert_2\notag\\
	&\conp 0,
\end{align}
where the last line is by Lemma A.2 of \cite{JT2023} and the CMT. 
From Lemma B.3 of \cite{JT2023}, 
$\Vert\hat{\theta}-\theta_0\Vert=O_p(m_0^{1/2}\lambda)=O_p\left(\left(\frac{m_0\log p}{n}\right)^{1/2}\right)=o_p(1)$.
Since $\dot{\rho}(\theta)$ is locally Lipschitz in a neighborhood of $\theta_0$, 
with probability approaching 1, we have for $B_0=O(1)$ 
$\Vert\dot{\rho}(\hat{\theta})-\dot{\rho}({\theta}_0)\Vert \leq B_0\Vert \hat{\theta}-{\theta}_0\Vert$.
Thus, 
\begin{equation}\label{eq: rhohat order}
	\Vert\dot{\rho}(\hat{\theta})-\dot{\rho}({\theta}_0)\Vert_2\leq {r}^{1/2}\Vert\dot{\rho}(\hat{\theta})-\dot{\rho}({\theta}_0)\Vert=O_p\left(\left(\frac{m_0\log p}{n}\right)^{1/2}\right).
\end{equation} 
By the triangle inequality and \eqref{eq: rhohat order}
\begin{equation}\label{eq: rhohat order2}
	\Vert\dot{\rho}(\hat{\theta})\Vert_2\leq \Vert\dot{\rho}(\hat{\theta})-\dot{\rho}({\theta}_0)\Vert_2
	+\Vert \dot{\rho}({\theta}_0)\Vert_2=O_p(1).
\end{equation}
Therefore, the quantity in \eqref{eq: HIH negl} is $o_p(1)$.
Consider the second term on the right-hand side of \eqref{eq: db tr ineq}. By the triangle inequality and \eqref{eq: rhohat order}, 
\begin{align*}
	&\Vert \dot{\rho}(\hat{\theta})'{H}({\theta}_0)^{-1}{I}({\theta}_0){H}({\theta}_0)^{-1}(\dot{\rho}(\hat{\theta})-\dot{\rho}(\theta_0))\Vert_2\\
	&\leq \Vert \dot{\rho}(\hat{\theta})\Vert_2\Vert{H}({\theta}_0)^{-1}{I}({\theta}_0){H}({\theta}_0)^{-1}\Vert_2\Vert \dot{\rho}(\hat{\theta})-\dot{\rho}(\theta_0)\Vert_2\\
	&\conp 0.
\end{align*}
Similarly, for the third term on the right-hand side of \eqref{eq: db tr ineq}
\begin{align*}
	&\Vert (\dot{\rho}(\hat{\theta})-\dot{\rho}(\theta_0))'{H}({\theta}_0)^{-1}{I}({\theta}_0){H}({\theta}_0)^{-1}\dot{\rho}(\theta_0)\Vert_2\\
	&\leq \Vert \dot{\rho}(\hat{\theta})-\dot{\rho}(\theta_0)\Vert_2\Vert{H}({\theta}_0)^{-1}{I}({\theta}_0){H}({\theta}_0)^{-1}\Vert_2\Vert\dot{\rho}(\theta_0)\Vert_2\\
	&\conp 0.
\end{align*}
\paragraph*{Step 2: $n^{1/2}\dot{\rho}(\hat{\theta})'\hat{H}(\hat{\theta})^{-1}S(\theta_0)-n^{1/2}\dot{\rho}(\theta_0)'{H}({\theta}_0)^{-1}S(\theta_0)=o_p(1)$.}~\\
Remark that from Assumption \ref{A: AsyValid}, $|\dot{g}(y_i, x_i' \theta_0)| \leq C_u$, $|w_i|\leq C_u$ and $\Vert x_i\Vert^2\leq (p+1)C_u^2$ a.s. for all $i$. 
Using the independence assumption, 
\begin{align*}
	\E[\Vert S(\theta_0)\Vert_2^2]=\E[\Vert S(\theta_0)\Vert^2]=
	n^{-2}\E\left[\sum_{i=1}^nw_i^2\Vert x_i\Vert^2 \dot{g}(y_i, x_i' \theta_0)^2\right]\leq n^{-1}(p+1)C_u^6.
\end{align*}
By Markov's inequality,
\begin{equation}\label{eq: score order}
	\Vert S(\theta_0)\Vert_2=O_p\left(\sqrt{\frac{p}{n}}\right).
\end{equation}
Now rewrite 
\begin{align}
	&n^{1/2}\dot{\rho}(\hat{\theta})'\hat{H}(\hat{\theta})^{-1}S(\theta_0)-n^{1/2}\dot{\rho}(\theta_0)'{H}({\theta}_0)^{-1}S(\theta_0)\notag\\
	&=n^{1/2}(\dot{\rho}(\hat{\theta})-\dot{\rho}(\theta_0))'\hat{H}(\hat{\theta})^{-1}S(\theta_0)+ n^{1/2}\left(\dot{\rho}(\theta_0)'\hat{H}(\hat{\theta})^{-1}S(\theta_0)-\dot{\rho}(\theta_0)'{H}({\theta}_0)^{-1}S(\theta_0)\right).
	\label{eq: rem}
\end{align}
For the first term of \eqref{eq: rem}, 
\begin{align}
	\Vert n^{1/2}(\dot{\rho}(\hat{\theta})-\dot{\rho}(\theta_0))'\hat{H}(\hat{\theta})^{-1}S(\theta_0)\Vert_2
	&\leq n^{1/2}\Vert \dot{\rho}(\hat{\theta})-\dot{\rho}(\theta_0)\Vert_2\Vert\hat{H}(\hat{\theta})^{-1}\Vert_2\Vert S(\theta_0)\Vert_2\notag\\
	&=n^{1/2}O_p\left(\sqrt{\frac{m_0\log p}{n}}\right)O_p(1)\,O_p\left(\sqrt{\frac{p}{n}}\right)\notag\\
	&=O_p\left(\sqrt{\frac{p\,m_0\log p}{n}}\right)\notag\\
	&=o_p(1),\label{eq: db step3b}
\end{align}
where the first inequality is by Cauchy-Schwarz, the first equality uses 
\begin{equation*}\label{eq: hess Op1}
	\Vert \hat{H}(\hat{\theta})^{-1}\Vert_2=O_p(1).
\end{equation*}
as shown in the proof Lemma 3.5 of \cite{JT2023}, \eqref{eq: rhohat order} and \eqref{eq: score order}, and 
the last equality holds because $m_0(\log p)p/n\leq m_0(\log p) (p/n)^{1/2}(p^2/n)^{1/2}\to 0$  by the assumption of the proposition. For the second term of \eqref{eq: rem}, we have 
\begin{align}
	n^{1/2}\Vert\dot{\rho}(\theta_0)'\hat{H}(\hat{\theta})^{-1}S(\theta_0)-\dot{\rho}(\theta_0)'{H}({\theta}_0)^{-1}S(\theta_0)\Vert_2
	&\leq 
	{n}^{1/2} \Vert \dot{\rho}(\theta_0) \Vert_2 \Vert \hat{H}(\hat{\theta})^{-1}-{H}({\theta}_0)^{-1}\Vert_2\Vert S(\theta_0)\Vert_2\notag\\
	&=n^{1/2}
	O_p\left(\sqrt{\frac{p}{n}} + m_0 \lambda \right)O_p\left(\sqrt{\frac{p}{n}}\right)\notag\\
	&=
	O_p\left(\sqrt{\frac{p^2}{n}} + \sqrt{p}\,m_0 \lambda \right)\notag\\
	&=o_p(1),\label{eq: db step3c}
\end{align}
where the first inequality is by Cauchy-Schwarz, the first equality is by Lemma A.2 of \cite{JT2023} and  \eqref{eq: score order}, and the last equality holds because $p^2/n\to 0$ and 
$p^{1/2}\,m_0 \lambda=Cm_0 (p/n)^{1/2}(\log p)^{1/2}\leq Cm_0 (p/n)^{1/2} 2\log p\to 0$ by the assumption of the proposition. It follows from \eqref{eq: rem}, \eqref{eq: db step3b} and 
\eqref{eq: db step3c} that 
$$n^{1/2}\dot{\rho}(\hat{\theta})'\hat{H}(\hat{\theta})^{-1}S(\theta_0)-n^{1/2}\dot{\rho}(\theta_0)'{H}({\theta}_0)^{-1}S(\theta_0)=o_p(1).$$
\paragraph*{Step 3: $n^{1/2}\dot{\rho}(\hat{\theta})'\hat{H}(\hat{\theta})^{-1}R=o_p(1)$.}~\\
By Cauchy-Schwarz, $n^{1/2}\Vert\dot{\rho}(\hat{\theta})'\hat{H}(\hat{\theta})^{-1}R \Vert_2\leq {n}^{1/2} \Vert\dot{\rho}(\hat{\theta})\Vert_2\Vert \hat{H}(\hat{\theta})^{-1}R\Vert_2$. Remark from \eqref{eq: rhohat order2} that $\Vert\dot{\rho}(\hat{\theta})\Vert_2=O_p(1)$. To show ${n}^{1/2}\Vert \hat{H}(\hat{\theta})^{-1}R\Vert_2 =o_p(1)$, note that 
\begin{align}
	\max_{1\leq j\leq p+1}\vert R_j\vert  
	&\leq
	{n}^{-1}\sum_{i=1}^n  \vert \ddot{g}(y_i, x_i'\theta^*) - \ddot{g}(y_i, x_i'\hat{\theta})\vert |w_i| \max_{1\leq j\leq p+1}|x_{ij}| | x_i'(\theta_0 - \hat{\theta})|\notag\\
	& \leq {n}^{-1} \sum_{i=1}^n L_0| x_i(\theta^*-\hat{\theta})| C_u^2 | x_i'(\theta_0 - \hat{\theta})|\notag\\
	& \leq L_0  C_u^2{n}^{-1} \sum_{i=1}^n  | x_i'(\theta_0 - \hat{\theta})|^2\notag\\ 
	&=L_0C_u^2O_p(m_0 \lambda^2)\notag\\
	&=O_p(m_0\lambda^2),\label{eq: Rmax}
\end{align}
where the first inequality is by Assumption \ref{A: AsyValid}\ref{AsyValid Lip}, and the first equality uses Lemma B.3 of \cite{JT2023}. 
Since $\Vert{H}(\theta_0)\Vert={O}(1)$ and $\Vert\hat{H}(\hat{\theta}) -{H}(\theta_0)\Vert=o_p(1)$, $\Vert\hat{H}(\hat{\theta})\Vert=O_p(1)$. Therefore, 

\begin{align}
	{n}^{1/2}\Vert\hat{H}(\hat{\theta})^{-1}R\Vert_2 
	& \leq {n}^{1/2}\Vert\hat{H}(\hat{\theta} )^{-1}\Vert_2\Vert R \Vert_2\notag\\
	& \leq{n}^{1/2}\hat{H}(\hat{\theta})^{-1} (p+1)^{1/2} \Vert R\Vert_{\infty}\notag\\
	&= O_p((n(p+1))^{1/2} m_0 \lambda^2)\notag\\
	&=o_p(1),
\end{align}
where the first equality holds by using \eqref{eq: Rmax} and the second equality follows on noting that 
$(n(p+1))^{1/2}m_0\lambda^2 = (n(p+1))^{1/2}m_0C^2(\log p)/n\leq (2p/n)^{1/2}m_0C^2\log p=o(1)$.

\section{Survey Questions Used for Variable Construction}\label{sec:AppendixQuestions}

This appendix provides the survey questions used to construct the Cyber Score, the $k$-means clustering variables, and the Business Digital Usage Score (BDUS). Each set of questions corresponds to specific aspects of digital technology adoption and cyber security challenges.

\subsection{Cyber Security Incidence Variable Construction}
The Cyber Security Incidence variable is equal to 1 if a firm answers ``Yes'' to one of the following questions and 0 otherwise.
\begin{quote}
	\emph{To the best of your knowledge, which cyber security incidents impacted your business in 2021? Select all that apply.}
\end{quote}

\begin{itemize}
	\item Incidents to disrupt or deface the business or web presence.
	\item Incidents to steal personal or financial information.
	\item Incidents to steal money or demand ransom payment.
	\item Incidents to steal or manipulate intellectual property or business data.
	\item Incidents to access unauthorised or privileged areas.
	\item Incidents to monitor and track business activity.
	\item Incidents with an unknown motive.
\end{itemize}

\subsection{Questions Used for $k$-means Clustering}
The $k$-means clustering variables were derived from responses to the following survey questions, which identify challenges businesses face in utilizing various digital and financial technologies. Each affirmative response indicates an inefficiency or challenge:
\begin{itemize}
	\item Does your business face challenges with online transaction processing?
	\item Does your business face challenges with digital marketing?
	\item Does your business face challenges with data analytics?
	\item Does your business face challenges with integrating digital technologies into business operations?
	\item Does your business face challenges with big data?
	\item Does your business face challenges with artificial intelligence?
	\item Does your business face challenges with cloud computing?
	\item Does your business face challenges with ICT infrastructure?
	\item Does your business face challenges with government connectivity?
	\item Does your business face challenges with website operations?
\end{itemize}

\subsection{Questions Used for Business Digital Usage Score (BDUS)}
The BDUS is based on responses to questions about the adoption of specific digital technologies. A ``Yes'' response to any of the following indicates that the business has utilized the respective technology:
\begin{itemize}
	\item Does your business use online transaction processing systems?
	\item Does your business use digital marketing platforms?
	\item Does your business use data analytics tools?
	\item Does your business integrate digital technologies into business operations?
	\item Does your business utilize big data technologies?
	\item Does your business employ artificial intelligence tools?
	\item Does your business use cloud computing services?
	\item Does your business maintain ICT infrastructure?
	\item Does your business interact with government systems digitally?
	\item Does your business operate its own website?
\end{itemize}
The BDUS score is computed as the total number of ``Yes'' responses to these questions, with higher scores indicating greater digital engagement.
\end{document}